\documentclass[12pt,reqno]{amsart}
\usepackage{graphicx,amscd, color,amsmath,amsfonts,amssymb,geometry,amssymb, xfrac,nicefrac}
\usepackage[initials]{amsrefs}

\newtheorem{theorem}{Theorem}[section]

\newtheorem{proposition}[theorem]{Proposition}
\newtheorem{corollary}[theorem]{Corollary}

\newtheorem{remark}[theorem]{Remark}

\newtheorem{lemma}[theorem]{Lemma}

\newtheorem{definition}[theorem]{Definition}

\numberwithin{equation}{section}

\newcommand{\rank}{\operatorname{rank}}
\newcommand{\Ima}{\operatorname{Im}}

\begin{document}

\title{On the structure of Completely Reducible States}

\author[Cariello ]{D. Cariello}

\address{Faculdade de Matem\'atica, \newline\indent Universidade Federal de Uberl\^{a}ndia, \newline\indent 38.400-902  Uberl\^{a}ndia, Brazil.}
\email{dcariello@ufu.br}


\subjclass[2010]{}

\begin{abstract}
The complete reducibility property for bipartite states reduced the separability problem to a proper subset of positive under partial transpose states and was used to prove several  theorems inside and outside entanglement theory. So far only three types of bipartite states were proved to possess this property. In this work, we provide some procedures to create states with this property, we call these states by the name of completely reducible states. The convex combination of such states is the first procedure, showing that the set of  completely reducible states is a convex cone.    We also provide a complete description of the extreme rays of this set. Then we show that powers, roots and partial traces of completely reducible states result in states of  the same type. Finally, we consider a shuffle of states that preserves this  property. This shuffle allows us to construct states with the complete reducibility property avoiding the only three  conditions known to date that imply this property. We conclude this paper by showing a connection between our results and the distillability problem. All these results were possible due to a simplification of the description of this property also presented here for the first time.
\end{abstract}

\maketitle

\section{Introduction}

\vspace{0,5cm}

There are many theorems concerning matrices with non-negative coefficients. One particularly important example is the Perron-Frobenius theorem \cite{Meyer}, which relates properties on the spectrum of such matrices with the  existence or non-existence of invariant subspaces generated by subsets of the canonical basis.\vspace{0,3cm}

This theorem was extended to linear operators acting on matrix algebras leaving the set of positive semidefinite Hermitian matrices invariant, the so-called positive maps \cite{evans}.  For example, the notion of irreducibility borrowed from Perron-Frobenius theory  was translated to positive maps  as follows:  Let $V\in\mathcal{M}_k$ be an orthogonal projection and consider the sub-algebra $V\mathcal{M}_kV=\{VXV, X\in\mathcal{M}_k\}$. A positive map $T:V\mathcal{M}_kV\rightarrow V\mathcal{M}_kV$ is called irreducible if the only sub-algebras $W\mathcal{M}_kW$ of $V\mathcal{M}_kV$  left invariant by $T$ satisfy $W=0$ or $W=V$.   \vspace{0,3cm}

In this work, we are not only interested in the case $V=Id_{k\times k}$. The main problem discussed here is related  to the construction of positive maps that are irreducible on some proper algebras ($V\neq Id_{k\times k}$) with additional conditions (See the definition of completely reducible map below).  \vspace{0,3cm}

Another  result concerning square non-negative matrices, which was also adapted to positive maps, is the Sinkhorn-Knopp theorem \cite{Sinkhorn}. This result states whether such matrix can be turned into a double stochastic one via multiplication on the left and on the right by positive diagonal matrices. The notions of support and total support play the key role in this theorem and surprisingly these concepts can be easily extended to positive maps (even rectangular ones) and analogous theorems follow for rectangular positive maps by just adapting Sinkhorn and Knopp original proofs  \cite{gurvits2004, CarielloLAMA}.  This extension to positive maps is used to study the existence of the filter normal form for bipartite states  \cite{gurvits2004,CarielloLAMA}. This normal form has been used to unify some separability conditions \cite{Git}.
\vspace{0,3cm}

Another notion originated from matrix theory  adapted to  positive maps,  which turned out to be exceptionally important in entanglement theory, is the complete reducibility property \cite[Definition 8.5, Ch I]{schaefer}. This property, owned by special types of states in this theory,  reduces the separability problem to a smaller subset of states (\cite[Corollary 14]{CarielloIEEE}). Its adaptation to positive maps can be described as follows.  Let $T:V\mathcal{M}_kV\rightarrow V\mathcal{M}_kV$ be a self-adjoint positive map with respect to the trace inner product  $($i.e., $tr(T(X)Y^*)=tr(XT(Y)^*)$, where $tr(Z)$ stands for the trace of $Z)$. We say that $T:V\mathcal{M}_kV\rightarrow V\mathcal{M}_kV$  is completely reducible if there are orthogonal projections $W_1,\ldots,W_l$ in $\mathcal{M}_k$ satisfying
\begin{itemize}
\item $W_i\mathcal{M}_kW_i\subset V\mathcal{M}_kV$ or, equivalently, $VW_i=W_iV=W_i$ for every $i$,
\item $W_iW_j=W_jW_i=0$ for $i\neq j$,
\item $\displaystyle V\mathcal{M}_kV=\bigoplus_{i=1}^lW_i\mathcal{M}_kW_i\oplus R$,
where $\displaystyle R$ is the orthogonal complement of $\displaystyle \bigoplus_{i=1}^lW_i\mathcal{M}_kW_i$ 

within $V\mathcal{M}_kV$ with respect to the trace inner product,
\item $T(W_i\mathcal{M}_kW_i)\subset W_i\mathcal{M}_kW_i$ for every $i$,
\item $T|_{W_i\mathcal{M}_kW_i}$ is irreducible for every $i$ and
\item $T|_{R}\equiv 0$\vspace{0,2cm}
\end{itemize}

 The only non-trivial requirement in the list above  is the last one, which is a very strong restriction as it implies, for example, that the identity map $(Id:\mathcal{M}_k\rightarrow \mathcal{M}_k)$ is not completely reducible. This last condition is also responsible for the reduction of the separability problem to a special subset  of states.\vspace{0,3cm}

Now, the easiest way to produce completely reducible maps is to add self-adjoint irreducible maps supported on orthogonal sub-algebras. However, there are indirect ways to produce these maps. It is surprising that entanglement oriented ideas play a  significant role in such a task. \vspace{0,3cm}

In order to understand how this property appears in entanglement theory, let us identify $\mathcal{M}_k\otimes \mathcal{M}_m\simeq \mathcal{M}_{km}$ via Kronecker product and consider $\gamma=\sum_{i=1}^nA_i\otimes B_i\in \mathcal{M}_k\otimes \mathcal{M}_m$. \vspace{0,3cm}

We say that $\gamma\in \mathcal{M}_k\otimes \mathcal{M}_m$ is a state (non-normalized) if $\gamma$ is a positive semidefinite Hermitian matrix. In addition, if its rank is 1 then we call it a pure state.  Now, for every given state $\gamma\in \mathcal{M}_k\otimes \mathcal{M}_m$, define  two linear maps $G_\gamma:\mathcal{M}_k\rightarrow \mathcal{M}_m$ and $F_\gamma:\mathcal{M}_m\rightarrow \mathcal{M}_k$ as

\vspace{0,3cm}

 \begin{center}
  $ G_{\gamma}(X)=\sum_{i=1}^n tr(A_iX)B_i$\  and\  $ F_{\gamma}(X)=\sum_{i=1}^n tr(B_iX)A_i.$
 \end{center}
  
  \vspace{0,3cm}

  The linear maps  -\ $ G_{\gamma}(X), F_{\gamma}(X) $\ - are positive maps as they send positive semidefinite Hermitian matrices into  positive semidefinite Hermitian matrices. This can be seen by the equation
\begin{equation}\label{eqadjoints}
tr(\gamma(X\otimes Y^*))=tr(G_\gamma(X)Y^*)=tr(XF_\gamma(Y^*))=tr(XF_{\gamma^*}(Y)^*)=tr(XF_\gamma(Y)^*).
\end{equation}

and the fact that whenever $\gamma,X,Y$ are positive semidefinite $tr(\gamma(X\otimes Y^*))\geq 0$. \vspace{0,3cm}

This equation also tells us that $G_{\gamma}$ and $F_{\gamma}$ are adjoints with respect to the trace inner product. Hence, $F_{\gamma}\circ G_{\gamma}:\mathcal{M}_k\rightarrow \mathcal{M}_k$ is a self-adjoint positive map. \vspace{0,3cm}

There are conditions to be imposed on a state $\gamma$ to guarantee the complete reducibility of $F_{\gamma}\circ G_{\gamma}:\mathcal{M}_k\rightarrow \mathcal{M}_k$. These conditions arise naturally in entanglement theory. For example,\vspace{0,3cm}

\begin{itemize}
\item[(i)] If $\gamma$ remains positive under partial transpose, i.e., $\gamma\geq 0$ and $\gamma^{\Gamma}=\sum_{i=1}^nA_i\otimes B_i^t\geq 0$, then $F_{\gamma}\circ G_{\gamma}:\mathcal{M}_k\rightarrow \mathcal{M}_k$ is completely reducible \cite[Theorem 26]{CarielloIEEE}.
\item[(ii)] If $\gamma$ remains positive under partial transpose composed with realignment, i.e., $\gamma\geq 0$ and $\mathcal{R}(\gamma^{\Gamma})\geq 0$, then $F_{\gamma}\circ G_{\gamma}:\mathcal{M}_k\rightarrow \mathcal{M}_k$ is completely reducible \cite[Theorem 27]{CarielloIEEE}. 
\item[(iii)] If $\gamma$ remains the same under realignment, i.e., $\gamma\geq 0$ and $\mathcal{R}(\gamma)=\gamma$, then $F_{\gamma}\circ G_{\gamma}:\mathcal{M}_k\rightarrow \mathcal{M}_k$ is completely reducible \cite[Theorem 28]{CarielloIEEE}.
\end{itemize}

\vspace{0,3cm}

The partial transpose and the realignment map are tools primarily used to detect entanglement \cite{peres, horodeckifamily,rudolph,rudolph2}. For us they are useful to construct  completely reducible maps indirectly.  The three types of states described above satisfy several  results \cite{CarielloArxiv} and many of these follow from the complete reducibility property. Below we describe some consequences of this property.

\begin{enumerate}
\item If $\gamma$ is one of these three types  of states (i-iii) and
\begin{itemize}
\item the non-zero Schmidt coefficients of $\gamma$ are equal then $\gamma$ is separable  \cite[Proposition 15]{CarielloIEEE}.
\item the rank of $\gamma$ coincides with its reduced ranks then $\gamma$ is separable \cite[Theorems 5.1, 5.2, 5.5]{CarielloArxiv}.
\end{itemize} 
\item If $\gamma$ remains positive under partial transpose (with one extra mild property) then the problem of determining whether $\gamma$ can be put in the filter normal form or not is reduced in polynomial time to the very specific problem of finding Perron eigenvectors of completely positive maps \cite{CarielloLMP}.
\item If $\mathbb{C}^{k}$ contains $k$ mutually unbiased bases then $\mathbb{C}^{k}$ contains $k+1$, which is known as Weiner's theorem \cite{weiner}. This follows from the first item above \cite[Theorem 36]{CarielloIEEE}.
\item There is also an application of this property  on entanglement breaking Markovian dynamics \cite{hanson}. 

\end{enumerate}

\vspace{0,2cm}

These consequences  should be more than enough to convince anyone of the importance of studying this property  more deeply. So our first major result is  the simplest description we could come up with for the complete reducibility property of $F_{\gamma}\circ G_{\gamma}:\mathcal{M}_k\rightarrow \mathcal{M}_k$ coming directly from the state $\gamma$.  In our corollary \ref{corollary2}, we show that  $F_{\gamma}\circ G_{\gamma}:\mathcal{M}_k\rightarrow \mathcal{M}_k$ is completely reducible if and only if for every pair of orthogonal projections $W\in\mathcal{M}_k$ and $V\in\mathcal{M}_m$, the following two conditions are equivalent
\begin{itemize}
\item[$a)$]  $tr(\gamma(W\otimes V^{\perp}))=tr(\gamma( W^{\perp}\otimes V))=0$,
\item[$b)$]  $\gamma=(W\otimes V)\gamma(W\otimes V)+(W^{\perp}\otimes V^{\perp})\gamma(W^{\perp}\otimes V^{\perp})$,
\end{itemize}
where $W^{\perp}=Id-W$ and $V^{\perp}=Id-V$.
\vspace{0,2cm}

It is obvious that $b)$ implies $a)$ for every state, but the converse is only true for states with this property.  From this new description, we deduce pretty much every other result of this article. 

\vspace{0,2cm}

Here we can get a glimpse of how this property actually works. For example, the conditions fulfilled by the states of item $(1)$ (described in the list above) imply the existence of orthogonal projections $W\in\mathcal{M}_k$ and $V\in\mathcal{M}_m$ satisfying the conditions of item $a)$ above. Then, by item $b)$, the states break into two pieces and inductions are performed on each piece proving the separability of these states.

\vspace{0,2cm}

The reduction of the separability problem to a proper subset of the positive under partial transpose states \cite[Corollary 14]{CarielloIEEE} follows the same line of reasoning. Whenever orthogonal projections $W\in\mathcal{M}_k$ and $V\in\mathcal{M}_m$ satisfying the conditions of item $a)$ exist for a positive under partial transpose state $\gamma$, they force this $\gamma$ to break into two pieces as described in item $b)$. The separability problem is now reduced to each of these pieces. The process can be repeated on each piece whenever new orthogonal projections satisfying the conditions of item $a)$ are found. 

\vspace{0,2cm}
Now, the separability problem can be reduced even further, to the set of states that remain positive under partial transpose composed with realignment, whose partial transposes are positive definite (See \cite{CarielloQIC3}).  Therefore, as argued in \cite{CarielloArxiv}, these three types of states (i-iii) should be treated  equally  within entanglement theory. Our attempt to do so is the definition presented in the next paragraph and the aim of this work described below.

\vspace{0,2cm}

For the sake of simplicity, let us call a state $\gamma\in \mathcal{M}_k\otimes \mathcal{M}_m$ such that $F_{\gamma}\circ G_{\gamma}:\mathcal{M}_k\rightarrow \mathcal{M}_k$ is completely reducible by the name of completely reducible state and let us define $$CR_{k, m}=\{\gamma\in \mathcal{M}_k\otimes\mathcal{M}_m, \gamma\geq 0 \text{ and } F_{\gamma}\circ G_{\gamma}:\mathcal{M}_k\rightarrow \mathcal{M}_k\text{ is completely reducible}\}.$$ 

\vspace{0,2cm}

The aim of this work is to study  $CR_{k,m}$ in order to gain more information on this set and to prove results in a more general way than those already proven for the triad of quantum states described above.  We also look for other types of matrices within this set besides the three types above. 

\vspace{0,2cm}

Our main investigative tool is a brand new description of the complete reducibility property (Corollary \ref{corollary2}).  Intuitively, this tool shows that whenever specific subspaces of $\mathbb{C}^k\otimes\mathbb{C}^m$ belong to the kernel of a state of $CR_{k,m}$ (See item $a)$ above), then that state decomposes nicely (See item $b)$ above). It is worth noticing that the absence of these specific subspaces within the kernel of states does not exclude these states from $CR_{k,m}$. Therefore there are trivial types of states in  $CR_{k,m}$, e.g., the positive definite type (Corollary \ref{corollaryCRdense}).

\vspace{0,2cm}
 
Now, we could simply exclude positive definite states from $CR_{k,m}$, in order to avoid this trivial case, but this would force $CR_{k,m}$ to lose one of its interesting properties, its convexity, which is another consequence of our new tool (Theorem \ref{theoremconvex}). 

\vspace{0,2cm}

Hence, we keep positive definite states inside $CR_{k,m}$ to show that this set is a dense convex cone inside the set of states of $\mathcal{M}_k\otimes \mathcal{M}_m$ (Theorem \ref{theoremconvex}). Since the set of states is closed and is distinct to $CR_{k,m}$ \cite[Lemmas 29, 30]{CarielloIEEE}, the cone $CR_{k,m}$ cannot  be closed. 

\vspace{0,2cm}

In particular, $CR_{k,m}$  cannot be the union of the three closed sets formed by the three types of states described in (i-iii). Therefore, there are states within $CR_{k,m}$ but outside these three closed sets. In the final part of our paper we provide an explicit method to create such states (Corollary \ref{corollary_new_types}) avoiding trivial examples such as these positive definite states. The states created by this method are by no means trivial, since we were only able to show their complete reducibility property by the new refined version of it.

\vspace{0,3cm}

Next, we also describe the extreme rays of $CR_{k,m}$  (Proposition \ref{propositionextremerays}).
It is quite remarkable that the cone of the separable states and the cone of the completely reducible states share the same extreme rays, despite the cone of the separable states being much smaller than $CR_{k,m}$  ($CR_{k,m}$ is dense within the cone of all states).

\vspace{0,3cm}

Then we show that powers, roots and partial traces of completely reducible states result in states of  the same type (See Propositions \ref{prop_power_roots_compred} and \ref{proppartialtraceiscompred}). These new results were completely out of reach without the new description of the complete reducibility property.

\vspace{0,3cm}

Finally, we define a shuffle of states (See Definition \ref{definitionshuffle}) that preserves the complete reducibility property whenever the original states are completely reducible (Theorem \ref{theoremshuffle}). We also notice that shuffling states of different types produce completely reducible states of different types, avoiding the three original requirements described above in  (i-iii) (Corollary \ref{corollary_new_types}).

\vspace{0,3cm}

Although, shuffling states of the same type (i-iii) preserve the type (See Theorem \ref{theoremshuffle2}). This property of the shuffle might be useful to approach  one of the most important open problems in quantum information theory: the distillability problem (See \cite{fiveproblems,horodeckifamily_distillability, Clarisse,LinChen}). 

\vspace{0,3cm}

The distillability problem can be simply stated as whether it is possible to shuffle $n$ times  the partial transpose of an entangled bipartite state $\gamma\in \mathcal{M}_k\otimes \mathcal{M}_k$  , $S(\gamma^{\Gamma},\ldots,\gamma^{\Gamma})\in \mathcal{M}_{k^{n}}\otimes \mathcal{M}_{k^{n}}$,  in order to find a vector $v=a\otimes b+c\otimes d\in \mathbb{C}^{k^{n}}\otimes \mathbb{C}^{k^{n}}$ such that $$tr(S(\gamma^{\Gamma},\ldots,\gamma^{\Gamma})vv^*)<0.$$ 

It is clear from the definition of the shuffle  that if $\gamma$ is a positive under partial transpose state then the answer is negative, but for states not possessing this property no definite answer has been found. As a consequence of our theorem 
\ref{theoremshuffle2},  we show that it is impossible to find such $v$ if $\dim(\text{span}\{a,b,c,d\})=2$ and $\gamma$ is invariant under realignment (See Corollary \ref{corollary_distillability}). We know that there are invariant under realignment states that are not positive under partial transpose (See \cite[Examples 3.36]{Cariello_thesis}).

\vspace{0,5cm}

This work is organized as follows. In section 2, we prove some preliminary results. In section 3, we show that the complete reducibility property for states can be described in a simpler way (Corollary \ref{corollary2}), which is used in the proof of the convexity of $CR_{k,m}$ (Theorem \ref{theoremconvex}). We also show that $CR_{k,m}$ is dense within the set of states (Corollary \ref{corollaryCRdense}).  In section 4,  we describe the extreme rays of $CR_{k,m}$.
In section 5, we define the shuffle of states, which is an operation that preserves the complete reducibility property. Other operations that preserve this property are also presented in this section.
In section 6, we construct completely reducible states avoiding the three known requirements described above in  (i-iii) and present the last result concerning the distillability problem (Corollary \ref{corollary_distillability}).
\section{Preliminaries}

\vspace{0,5cm}

Let $V,W\in \mathcal{M}_k$ be orthogonal projections, define $W\mathcal{M}_k=\{WX,\ X\in \mathcal{M}_k\}$, $\mathcal{M}_kW=\{XW,\ X\in \mathcal{M}_k\}$ and $V\mathcal{M}_kW=\{VXW,\ X\in \mathcal{M}_k\}$. Also, consider $W^{\perp}=Id-W$ and let $P_k$ denote the set of positive semidefinite Hermitian matrices of order $k$.  Whenever a matrix $\gamma\in P_k$, we shall denote this fact by $\gamma\geq 0$.

\vspace{0,2cm}

 We present here some basic results on completely positive maps and one basic result on completely reducible maps as defined in the introduction. Let us start with the basics.

\vspace{0,2cm}

\begin{definition}A linear map $T:\mathcal{M}_k\rightarrow\mathcal{M}_m$ is completely positive if there are matrices $R_1,\ldots,R_s\in \mathcal{M}_{m\times k}$ such that $T(X)=\sum_{i=1}^sR_iXR_i^*.$\end{definition}

\vspace{0,5cm}

The following theorem due to Choi \cite{choi} characterizes the completely positive maps.
\begin{theorem}\label{theoremchoi} Let $u_k=\sum_{i=1}^ke_i\otimes e_i$, where $e_1,\ldots,e_k$ is the canonical basis of $\mathbb{C}^k$. Then the linear map  $T:\mathcal{M}_k\rightarrow \mathcal{M}_m$ is completely positive if and only if $(Id\otimes T)(u_ku_k^*)$ is a state, or equivalently, $( T\otimes Id)(u_ku_k^*)$ is a state.\end{theorem}
\vspace{0,3cm}

\begin{remark}\label{remarkpropertiesadjoints}
From the definition of $G_{\gamma}$ and $F_{\gamma}$ described in the introduction, it can be seen that \begin{equation}\label{eq1} \gamma=F_{\gamma}((\cdot)^t)\otimes Id\ (u_mu_m^*)=Id\otimes G_{\gamma}((\cdot)^t)\ (u_ku_k^*) 
\end{equation}
 $($See \cite[Propositions 5 and 8]{cariello}$)$. 

Now, if we choose $R\in\mathcal{M}_{k\times m}$ and define  $\gamma=(R\otimes Id) (u_mu_m^*) (R^*\otimes Id)\in \mathcal{M}_k\otimes \mathcal{M}_m$ then 
\begin{itemize}
\item $\gamma$ is also equal to $(Id\otimes R^t) (u_ku_k^*) (Id\otimes \overline{R})$ and
\item  $F_{\gamma}(X^t)=RXR^*$ and $G_{\gamma}(X^t)=R^tX\overline{R}$ by the equation \eqref{eq1}.
\end{itemize}

In particular, if $R=Id$ then $\gamma=u_ku_k^*$ and $G_{u_ku_k^*}(X^t)=F_{u_ku_k^*}(X^t)=X$.
\end{remark}

\vspace{0,3cm}

\begin{lemma}\label{lemmaCPselfadjoint}Let $\gamma\in \mathcal{M}_k\otimes\mathcal{M}_m$ be a state. Then  $F_{\gamma}\circ G_{\gamma}:\mathcal{M}_k\rightarrow \mathcal{M}_k$ is a completely positive map and self-adjoint  with respect to the trace inner product.\end{lemma}
\begin{proof}Let $u_m\in\mathbb{C}^m\otimes\mathbb{C}^m$ and $u_k\in\mathbb{C}^k\otimes\mathbb{C}^k$ be as in theorem \ref{theoremchoi}.\vspace{0,3cm}

As explained in remark \ref{remarkpropertiesadjoints},  $$\gamma=F_{\gamma}((\cdot)^t)\otimes Id\ (u_mu_m^*)=Id\otimes G_{\gamma}((\cdot)^t)\ (u_ku_k^*).$$

By Choi's theorem, since $\gamma$ is posititive semidefinite, the maps $G_{\gamma}((\cdot)^t)$ and  $F_{\gamma}((\cdot)^t)$ are completely positive. Therefore, the inner and outer composition of $F_{\gamma}((\cdot)^t)$ with transposition, i.e.  $F_{\gamma}(\cdot)^t$, is also completely positive. 
\vspace{0,3cm}

Hence $F_{\gamma}(G_{\gamma}((\cdot)^t)^t$ is completely positive. Again, the inner and outer composition of  this map  with transposition gives the following completely positive map $F_{\gamma}(G_{\gamma}(\cdot)).$\vspace{0,3cm}

As discussed earlier, $G_{\gamma}$ and $F_{\gamma}$ are adjoints with respect to the trace inner product by equation \eqref{eqadjoints}. Therefore, $F_{\gamma}(G_{\gamma}(\cdot))$ is a completely positive self-adjoint map.\end{proof}

\vspace{0,5cm}

We shall also need the following well known lemma. 

\begin{lemma}\label{lemmanullcompletelypositive} Let $T:\mathcal{M}_k\rightarrow \mathcal{M}_k$ be a completely positive map and $\gamma\in P_k$ such that $T(\gamma)=0$. Then $T|_{W\mathcal{M}_k+\mathcal{M}_kW}\equiv 0,$ where  $W\in M_k$ is the orthogonal projection onto $\Ima(\gamma)$.
\end{lemma} 
\begin{proof}Since $T:\mathcal{M}_k\rightarrow \mathcal{M}_k$ is completely positive, there are matrices $R_1,\ldots,R_s\in \mathcal{M}_{k}$ such that $T(X)=\sum_{i=1}^sR_iXR_i^*$. \vspace{0,3cm}

Now, $\Ima(T(\gamma))=\Ima(T(W))$, because $T$ is a positive map and $\Ima(\gamma)=\Ima(W)$. \vspace{0,3cm}

Hence $T(W)=0$ and $R_iWR_i^*=0$ for every $i$.  Notice that $R_iWR_i^*=R_iW(R_iW)^*$, so $R_iW=WR_i^*=0$ for every $i$.   Finally, $R_iWXR_i^*=R_iYWR_i^*=0$ for every $i$.
\end{proof} 

\vspace{0,3cm} 

 For the following sections we shall need the next two lemmas.

\vspace{0,3cm} 

\begin{lemma}\label{lemmasubalgebrainvariant}Let $\gamma\in\mathcal{M}_k\otimes\mathcal{M}_m$ be a state and $W\in M_k$ be an orthogonal projection. Then $F_{\gamma}\circ G_{\gamma}:\mathcal{M}_k\rightarrow \mathcal{M}_k$ leaves the sub-algebra $W\mathcal{M}_kW$ invariant if and only if there is an orthogonal projection $V\in\mathcal{M}_m$ such that $tr(\gamma(W\otimes V^{\perp}))=tr(\gamma(W^{\perp}\otimes V))=0$.\end{lemma}
\begin{proof} 
 First, notice that, by equation \eqref{eqadjoints}, \begin{center}
 $tr(\gamma(W\otimes V^{\perp}))=tr(G_{\gamma}(W)V^{\perp})$ and $tr(\gamma(W^{\perp}\otimes V))=tr(W^{\perp}F_{\gamma}(V))$. 
 \end{center}
 
If we assume that $V$ is the orthogonal projection such that $\Ima(V)=\Ima(G_{\gamma}(W))$, where $W$ is the orthogonal projection satisfying $F_{\gamma}\circ G_{\gamma}(W\mathcal{M}_kW)\subset W\mathcal{M}_kW$, then, by the definition of $V$, $\Ima(V^{\perp})=\ker(G_{\gamma}(W))$. Hence 
 $tr(G_{\gamma}(W)V^{\perp})=0$.\vspace{0,3cm}
 
 Next, since $F_{\gamma}:\mathcal{M}_m\rightarrow \mathcal{M}_k$ is a positive map and $\Ima(V)=\Ima(G_{\gamma}(W))$ then $$\Ima(F_{\gamma}(V))=\Ima(F_{\gamma}(G_{\gamma}(W)))\subset \Ima(W).$$ 
 
 Hence $tr(F_{\gamma}(V)W^{\perp})=0$.\vspace{0,3cm}
 
 Now, for the converse,  $tr(G_{\gamma}(W)V^{\perp})=tr(W^{\perp}F_{\gamma}(V))=0$ imply that \begin{center}
 $\Ima(G_{\gamma}(W))\subset\Ima(V)$ and $\Ima(F_{\gamma}(V))\subset\Ima(W)$. 
 \end{center}
\vspace{0,3cm} 
 
 Thus, $\Ima(F_{\gamma}(G_{\gamma}(W)))\subset \Ima(W)$. Finally, since $F_{\gamma}\circ G_{\gamma}:\mathcal{M}_k\rightarrow \mathcal{M}_k$ is a positive map, it implies that $F_{\gamma}\circ G_{\gamma}(W\mathcal{M}_kW)\subset W\mathcal{M}_kW.$
\end{proof}

  \vspace{0,3cm}

\begin{lemma}\label{lemmatensorcompred}
Let $T:\mathcal{M}_k\rightarrow \mathcal{M}_k$ and $S:\mathcal{M}_m\rightarrow\mathcal{M}_m$ be completely reducible self-adjoint positive maps. Then $T\otimes S:\mathcal{M}_k\otimes\mathcal{M}_m\rightarrow \mathcal{M}_k\otimes\mathcal{M}_m$ is completely reducible too. By induction the same is true for the tensor product of any number of completely reducible maps.
\end{lemma}
\begin{proof}
Let  \begin{itemize}
\item $\mathcal{M}_k=\bigoplus_{i=1}^lW_i\mathcal{M}_kW_i\oplus R$ and $\mathcal{M}_m=\bigoplus_{i=1}^sV_i\mathcal{M}_mV_i\oplus R'$
\item $W_iW_j=0$ and $V_iV_j=0$ for $i\neq j$,
\item $\displaystyle R\perp \bigoplus_{i=1}^lW_i\mathcal{M}_kW_i$ and $\displaystyle R'\perp \bigoplus_{i=1}^sV_i\mathcal{M}_kV_i$ with respect to the trace inner product,
\item $T(W_i\mathcal{M}_kW_i)\subset W_i\mathcal{M}_kW_i$ and $S(V_i\mathcal{M}_kV_i)\subset V_i\mathcal{M}_kV_i$ for every $i$,
\item $T|_{W_i\mathcal{M}_kW_i}$ and $S|_{V_i\mathcal{M}_kV_i}$ are irreducible for every $i$ and
\item $T|_{R}\equiv 0$ and $S|_{R'}\equiv 0$. \vspace{0,2cm}
\end{itemize}

Now, since $T$ is self-adjoint and positive,  $T|_{W_i\mathcal{M}_kW_i}$ is irreducible if and only if the multiplicity of the spectral radius of $T|_{W_i\mathcal{M}_kW_i}$ is 1 and associated to the spectral radius there is an eigenvector $\gamma_i$ such that $\Ima(\gamma_i)=\Ima(V_i)$ and $\gamma_i\in P_k$ (See \cite[Lemma 5]{CarielloIEEE}). The same is true for $S|_{V_j\mathcal{M}_mV_j}$ and for the maps $T\otimes S|_{W_i\mathcal{M}_kW_i\otimes V_j\mathcal{M}_mV_j}$, for every $i,j$.
Hence $T\otimes S|_{W_i\mathcal{M}_kW_i\otimes V_j\mathcal{M}_mV_j}$ is irreducible for every $i,j$. \vspace{0,2cm}

Finally, notice that
\begin{itemize}
\item  $R\otimes\mathcal{M}_m+\mathcal{M}_k\otimes R'\subset \ker(T\otimes S)$ 
\item  $\displaystyle R\otimes\mathcal{M}_m+\mathcal{M}_k\otimes R'\perp \bigoplus_{i,j}W_i\mathcal{M}_kW_i\otimes V_j\mathcal{M}_mV_j$.
\item $\mathcal{M}_k\otimes\mathcal{M}_m=  \bigoplus_{i,j}W_i\mathcal{M}_kW_i\otimes V_j\mathcal{M}_mV_j\oplus (R\otimes\mathcal{M}_m+\mathcal{M}_k\otimes R').$  \vspace{0,2cm}
\end{itemize}

Hence 
$T\otimes S:\mathcal{M}_k\otimes\mathcal{M}_m\rightarrow \mathcal{M}_k\otimes\mathcal{M}_m$ is completely reducible. 
\end{proof}
 
\vspace{0,5cm} 
 
\section{Completely reducible states form a convex cone}

\vspace{0,5cm}

In this section we prove that the set of completely reducible states is a convex cone (theorem \ref{theoremconvex}). In the next section, we provide a complete description of the extreme rays of this cone (proposition \ref{propositionextremerays}).\vspace{0,3cm}
 
 In order to prove these results we need first the following one.
This lemma was proved in \cite[Proposition 7]{CarielloIEEE}.\vspace{0,3cm}

\begin{lemma}\label{lemmadecompositionproperty}Let $T:\mathcal{M}_k\rightarrow \mathcal{M}_k$ be a positive map and self-adjoint with respect to the trace inner product. Then $T$ is completely reducible if and only if 
for every  $\gamma\in P_k$ such that $T(\gamma)=\lambda\gamma$ and $\lambda>0$, we have
$T|_{W\mathcal{M}_kW^{\perp}+W^{\perp}\mathcal{M}_kW}\equiv 0,$
where  $W\in M_k$ is the orthogonal projection onto $\Ima(\gamma)$.
\end{lemma}

\vspace{0,3cm}

The next proposition is  a crucial step to obtain a simpler characterization of the complete reducibility property (corollary \ref{corollary2}), which is used in  the proof of the convexity of the set of completely reducible states (theorem \ref{theoremconvex}). It says that if $T$ in the previous lemma is completely positive then a stronger result holds.\vspace{0,3cm}

\begin{proposition}\label{propositionnull}Let $T:\mathcal{M}_k\rightarrow \mathcal{M}_k$ be a completely positive map and self-adjoint with respect to the trace inner product. Then $T$ is completely reducible if and only if  for every sub-algebra $W\mathcal{M}_kW$ left invariant by $T$, i.e.
$T(W\mathcal{M}_kW)\subset W\mathcal{M}_kW$, we have  $T|_{W\mathcal{M}_kW^{\perp}+W^{\perp}\mathcal{M}_kW}\equiv 0.$

\end{proposition}
\begin{proof}
First, let us assume that for every sub-algebra $W\mathcal{M}_kW$ left invariant by $T$, we have  $T|_{W\mathcal{M}_kW^{\perp}+W^{\perp}\mathcal{M}_kW}\equiv 0.$ Our goal is to show that $T$ is completely reducible.

\vspace{0,3cm}

Now,  consider $\gamma\in P_k$  such that $T(\gamma)=\lambda\gamma$ and $\lambda>0$. Let $W\in M_k$ be the orthogonal projection onto $\Ima(\gamma)$.
Since $T$ is a positive map and $T(\gamma)=\lambda\gamma$, we have  $T(W\mathcal{M}_kW)\subset W\mathcal{M}_kW.$

\vspace{0,3cm}

By assumption $T|_{W\mathcal{M}_kW^{\perp}+W^{\perp}\mathcal{M}_kW}\equiv 0$, which shows that $T$ is completely reducible by  lemma \ref{lemmadecompositionproperty}.
Next, let us prove the converse. So assume that $T$ is completely reducible.

\vspace{0,3cm}
Recall that $T$ being completely reducible means that $\displaystyle \mathcal{M}_k=\bigoplus_{i=1}^lW_i\mathcal{M}_kW_i\bigoplus R$, where 
\begin{itemize}
\item $W_iW_j=0$ for $i\neq j$,
\item $\displaystyle R\perp \bigoplus_{i=1}^lW_i\mathcal{M}_kW_i$ with respect to the trace inner product,
\item $T(W_i\mathcal{M}_kW_i)\subset W_i\mathcal{M}_kW_i$ for every $i$,
\item $T|_{W_i\mathcal{M}_kW_i}$ is irreducible and
\item $T|_{R}\equiv 0$ \vspace{0,3cm}
\end{itemize}

Since $T$ is self-adjoint and $T(W\mathcal{M}_kW)\subset W\mathcal{M}_kW$,
$$tr(T(W^{\perp})W)=tr(W^{\perp}T(W))=0.$$

Hence $T(W^{\perp}\mathcal{M}_kW^{\perp})\subset W^{\perp}\mathcal{M}_kW^{\perp}$.\vspace{0,3cm}

Now, by \cite[Lemma 6, Proposition 7]{CarielloIEEE}, $T|_{W\mathcal{M}_kW}$ and  $T|_{W^{\perp}\mathcal{M}_kW^{\perp}}$ are completely reducible. Thus, \begin{center}
$\displaystyle W\mathcal{M}_kW=\bigoplus_{i=1}^sW_i'\mathcal{M}_kW_i'\bigoplus R'$ and $\displaystyle W^{\perp}\mathcal{M}_kW^{\perp}=\bigoplus_{i=1}^t\widetilde{W_i}\mathcal{M}_k\widetilde{W_i}\bigoplus \widetilde{R},$
\end{center} where, for every $i$, 
$W_i'\mathcal{M}_kW_i'$, $\widetilde{W_i}\mathcal{M}_k\widetilde{W_i}$, $R'$ and $\widetilde{R}$ have the same characteristics of  $W_i\mathcal{M}_kW_i$ and $R$ as described in the items above.
\vspace{0,3cm}

Next, notice that $A=W-W_1'-\ldots-W_s'\in R'$ and $B=W^{\perp}-\widetilde{W}_1-\ldots-\widetilde{W}_s\in \widetilde{R}.$ Hence $T(A)=T(B)=0$.\vspace{0,3cm}

By lemma \ref{lemmanullcompletelypositive}, since $T$ is completely positive,  \begin{equation}\label{eq2}
T|_{A\mathcal{M}_k+\mathcal{M}_kA}\equiv 0 \text{ and } T|_{B\mathcal{M}_k+\mathcal{M}_kB}\equiv 0.
\end{equation}

Notice that $$W\mathcal{M}_kW^{\perp}=\bigoplus_{i,j} W_i'\mathcal{M}_k\widetilde{W}_j\bigoplus \left(A\mathcal{M}_kW^{\perp}+ W\mathcal{M}_kB\right),$$
$$W^{\perp}\mathcal{M}_kW=\bigoplus_{j,i} \widetilde{W}_j\mathcal{M}_kW'_i\bigoplus \left(B\mathcal{M}_kW+ W^{\perp}\mathcal{M}_kA\right).$$

Recall that $T|_{W_i'\mathcal{M}_kW_i'}$ and $T|_{\widetilde{W}_i\mathcal{M}_k\widetilde{W}_i}$ are irreducible, but the only sub-algebras left invariant by $T$, where $T$ is irreducible, are the sub-algebras $W_i\mathcal{M}_kW_i$, because $T$ is completely reducible (See \cite[Proposition 7]{CarielloIEEE}).\vspace{0,3cm}

So, every  $W_i'$ and every $\widetilde{W}_j$, must be equal to some $W_r$. Thus, 
every  $W_i'\mathcal{M}_k\widetilde{W}_j$ and $\widetilde{W}_j\mathcal{M}_kW'_i$ are orthogonal to $\bigoplus_{i=1}^lW_i\mathcal{M}_kW_i$.
Consequently, $W_i'\mathcal{M}_k\widetilde{W}_j\subset R$ and $\widetilde{W}_j\mathcal{M}_kW'_i\subset R$. Therefore $$\displaystyle \left(\bigoplus_{i,j} W_i'\mathcal{M}_k\widetilde{W}_j\right)\oplus \left(\bigoplus_{j,i} \widetilde{W}_j\mathcal{M}_kW'_i\right)\subset R\subset \ker(T).$$\vspace{0,3cm}

Recall that $A\mathcal{M}_kW^{\perp}+ W\mathcal{M}_kB+ B\mathcal{M}_kW+ W^{\perp}\mathcal{M}_kA\subset \ker(T)$ 
by equation \eqref{eq2}. So $W\mathcal{M}_kW^{\perp}+W^{\perp}\mathcal{M}_kW\subset \ker(T)$.
\end{proof}

\vspace{0,3cm}

\vspace{0,3cm} 

This last proposition  provides a way to check the complete reducibility of $F_{\gamma}\circ G_{\gamma}:\mathcal{M}_k\rightarrow \mathcal{M}_k$ directly from $\gamma$ as you can see in the next corollary. This is the simplest characterization of the complete reducibility property that we are aware. 

\vspace{0,3cm}

\begin{corollary}\label{corollary2} Let $\gamma\in\mathcal{M}_k\otimes\mathcal{M}_m$ be a state. Then $\gamma\in CR_{k,m}$ if and only if for every pair of orthogonal projections $W\in\mathcal{M}_k$ and $V\in\mathcal{M}_m$, the following two conditions are equivalent
\begin{itemize}
\item[$a)$]  $tr(\gamma(W\otimes V^{\perp}))=tr(\gamma( W^{\perp}\otimes V))=0$,
\item[$b)$]  $\gamma=(W\otimes V)\gamma(W\otimes V)+(W^{\perp}\otimes V^{\perp})\gamma(W^{\perp}\otimes V^{\perp})$.
\end{itemize}
\end{corollary}
\begin{proof}It is obvious that if $b)$ holds then $a)$ holds for every state $\gamma\in\mathcal{M}_k\otimes\mathcal{M}_m$.

\vspace{0,2cm}

Now, let us prove the converse  for states in $CR_{k,m}$. So let $\gamma\in CR_{k,m}$ and assume that  $$tr(\gamma(W\otimes V^{\perp}))=tr(\gamma(W^{\perp}\otimes V))=0.$$ 

Since $\gamma\geq 0$, $(W\otimes V^{\perp})\geq 0$ and $(W^{\perp}\otimes V)\geq 0$, from the equation above we get \begin{center}
$\gamma(W\otimes V^{\perp})=(W\otimes V^{\perp})\gamma=0$ and $\gamma( W^{\perp}\otimes V)=(W^{\perp} \otimes V)\gamma=0.$
\end{center}

\vspace{0,2cm}

Therefore, $\gamma=(W+W^{\perp}\otimes V+V^{\perp})\gamma (W+W^{\perp}\otimes V+V^{\perp})=$ \begin{equation*}
(W\otimes V)\gamma(W\otimes V)+(W^{\perp}\otimes V^{\perp})\gamma(W^{\perp}\otimes V^{\perp})+
\end{equation*}
\begin{equation*}
 (W\otimes V)\gamma(W^{\perp}\otimes V^{\perp})+(W^{\perp}\otimes V^{\perp})\gamma(W\otimes V).
\end{equation*}

\vspace{0,2cm}

It remains to show that 
\begin{equation}\label{eq4}
(W\otimes V)\gamma(W^{\perp}\otimes V^{\perp})+(W^{\perp}\otimes V^{\perp})\gamma(W\otimes V)=0.
\end{equation}

\vspace{0,2cm}

By lemma \ref{lemmasubalgebrainvariant}, the sub-algebra $W\mathcal{M}_kW$ is left invariant by $F_{\gamma}\circ G_{\gamma}:\mathcal{M}_k\rightarrow\mathcal{M}_k$.
Since $F_{\gamma}\circ G_{\gamma}:\mathcal{M}_k\rightarrow\mathcal{M}_k$ is completely positive, self-adjoint (lemma \ref{lemmaCPselfadjoint}) and completely reducible by hypothesis, we have $$F_{\gamma}\circ G_{\gamma}|_{W\mathcal{M}_kW^{\perp}+W^{\perp}\mathcal{M}_kW}\equiv 0,$$
by lemma \ref{propositionnull}.

\vspace{0,2cm}

Recall that $F_{\gamma}:\mathcal{M}_m\rightarrow \mathcal{M}_k$ is the adjoint of $G_{\gamma}:\mathcal{M}_k\rightarrow \mathcal{M}_m$, whence  \begin{equation}\label{eq5}
G_{\gamma}|_{W\mathcal{M}_kW^{\perp}+W^{\perp}\mathcal{M}_kW}\equiv 0.
\end{equation}

\vspace{0,2cm}

Let $u_k$ be as in theorem \ref{theoremchoi}. Notice that, $u_k=\sum_{i=1}^kv_i\otimes \overline{v_i}$, for any orthornomal basis $v_1,\ldots,v_k$ of $\mathbb{C}^k$. Hence,  we have $(W\otimes \overline{W}+W^{\perp}\otimes \overline{W^{\perp}})u_k=u_k.$

\vspace{0,2cm}

Thus,  $u_ku_k^*=(W\otimes \overline{W}+W^{\perp}\otimes \overline{W^{\perp}})u_ku_k^*(W\otimes \overline{W}+W^{\perp}\otimes \overline{W^{\perp}})=\sum_{i,j=1}^2(W_i\otimes \overline{W_i})u_ku_k^*(W_j\otimes \overline{W_j})$, where $W_1=W$ and $W_2=W^{\perp}$.

\vspace{0,2cm}

Recall that $\gamma=Id\otimes G_{\gamma}(\cdot)(u_ku_k^*)^{\Gamma}$ by remark \ref{remarkpropertiesadjoints}. So 

$$\gamma=Id\otimes G_{\gamma}(\cdot)(u_ku_k^*)^{\Gamma}=\sum_{i,j=1}^2Id\otimes G_{\gamma}(\cdot)((W_i\otimes W_j)(u_ku_k^*)^{\Gamma}(W_j\otimes W_i)).$$

By equation \eqref{eq5}, whenever $i\neq j$, we have $Id\otimes G_{\gamma}(\cdot)((W_i\otimes W_j)(u_ku_k^*)^{\Gamma}(W_j\otimes W_i))=0$. Therefore 
$$\gamma=\sum_{i=1}^2Id\otimes G_{\gamma}(\cdot)((W_i\otimes W_i)(u_ku_k^*)^{\Gamma}(W_i\otimes W_i)).$$
Moreover, this last equation implies that  $(W_i\otimes R)\gamma(W_j\otimes S)=0$, whenever $i\neq j$, for arbitrary matrices $R,S\in\mathcal{M}_m$, which proves the desired equation \eqref{eq4}.

\vspace{0,2cm}

We have just proved that if $\gamma\in CR_{k,m}$ then the two conditions described in the statement of this corollary are equivalent. It remains to show the converse. Let us assume the equivalence of the two conditions. Consider a sub-algebra $W\mathcal{M}_kW$ left invariant by $F_{\gamma}\circ G_{\gamma}:\mathcal{M}_k\rightarrow\mathcal{M}_k$. 

\vspace{0,2cm}

By lemma \ref{lemmasubalgebrainvariant}, there is an orthogonal projection $V\in \mathcal{M}_m$ such that $tr(\gamma(W\otimes V^{\perp}))=tr(\gamma( W^{\perp}\otimes V))=0$. Since the two conditions are equivalent, we have  $$\gamma=(W\otimes V)\gamma(W\otimes V)+(W^{\perp}\otimes V^{\perp})\gamma(W^{\perp}\otimes V^{\perp}).$$

Therefore, from the definition of $G_{\gamma}:\mathcal{M}_k\rightarrow \mathcal{M}_m$, it follows that $G_{\gamma}|_{W\mathcal{M}_kW^{\perp}+W^{\perp}\mathcal{M}_kW}\equiv 0$, whence $F_{\gamma}\circ G_{\gamma}|_{W\mathcal{M}_kW^{\perp}+W^{\perp}\mathcal{M}_kW}\equiv 0.$ By proposition \ref{propositionnull}, $F_{\gamma}\circ G_{\gamma}:\mathcal{M}_k\rightarrow\mathcal{M}_k$ is completely reducible, which means that $\gamma\in CR_{k,m}$  as desired. \end{proof}

\vspace{0,5cm}
 The previous corollary shows that positive under partial transpose states are completely reducible in a very straightforward way (See next corollary). In addition, it is very useful in the proof of the convexity of the set $CR_{k,m}$ (See theorem \ref{theoremconvex}).

\vspace{0,5cm}
\begin{corollary}
If $\gamma\in\mathcal{M}_k\otimes\mathcal{M}_m$ is a positive under partial transpose state then $\gamma\in CR_{k,m}$.
\end{corollary}
\begin{proof}
Let  $W_1\in\mathcal{M}_k$ and $V_1\in\mathcal{M}_m$ be orthogonal projections and define, $W_2=W_1^{\perp}$ and $V_2=V_1^{\perp}$.

 In addition assume that
 \begin{equation}\label{eq6}
tr(\gamma(W_1\otimes V_2))=tr(\gamma( W_2\otimes V_1))=0.
\end{equation} 

\vspace{0,2cm}

We must show that $\gamma=\sum_{i=1}^2(W_i\otimes V_i)\gamma(W_i\otimes V_i)$ to conclude that  $\gamma\in CR_{k,m}$ by corollary \ref{corollary2}.

\vspace{0,2cm}

By equation \eqref{eq6}, since $\gamma\geq 0$, $W_1\otimes V_2\geq 0$ and $W_2\otimes V_1\geq 0$, we get $\gamma(W_i\otimes V_j)=0$ for $i\neq j$. 

\vspace{0,2cm}

Thus,
\begin{equation}\label{eq7}
 \gamma=\sum_{i,j=1}^2(W_i\otimes V_i)\gamma(W_j\otimes V_j).
 \end{equation}

Now, equation \eqref{eq6} also implies that $tr(\gamma^{\Gamma}(W_1\otimes \overline{V_2}))=tr(\gamma^{\Gamma}(W_2\otimes \overline{V_1}))=0$. 

\vspace{0,2cm}

Since $\gamma^{\Gamma}\geq 0$, $W_1\otimes \overline{V_2}\geq 0$ and $W_2\otimes \overline{V_1}\geq 0$,  we  also get $\gamma^{\Gamma}(W_i\otimes \overline{V_j})=0$, whenever $i\neq j$. 

\vspace{0,2cm}

Thus, by equation \eqref{eq7}, we have 
$$\gamma^{\Gamma}=\sum_{i,j=1}^2(W_i\otimes \overline{V_j})\gamma^{\Gamma}(W_j\otimes \overline{V_i})=\sum_{i=1}^2(W_i\otimes \overline{V_i})\gamma^{\Gamma}(W_i\otimes \overline{V_i}).$$

Therefore $\gamma=\sum_{i=1}^2(W_i\otimes V_i)\gamma(W_i\otimes V_i)$, which completes the proof.
\end{proof}

\vspace{0,3cm}

The following corollary shows that positive definite states are trivially completely reducible. In the next theorem we see that $CR_{k,m}$ is a convex cone, so $CR_{k,m}$ is a dense convex cone inside the cone of all states. Despite being a large cone, we shall see in the next section  that its extreme rays coincide with the extreme rays of the cone of the separable states.
 
 \vspace{0,3cm}  
 
\begin{corollary}\label{corollaryCRdense}Every positive definite state of $\mathcal{M}_k\otimes \mathcal{M}_{m}$ belongs to $CR_{k,m}$. In particular, $CR_{k,m}$ is dense in the set of states and it is not closed.\end{corollary} 
 \begin{proof}

Let $\gamma\in  \mathcal{M}_k\otimes \mathcal{M}_{m}$ be a positive definite state.  Let  $W\in\mathcal{M}_k$ and $V\in\mathcal{M}_m$ be orthogonal projections satisfying
$$tr(\gamma(W\otimes V^{\perp}))=tr(\gamma( W^{\perp}\otimes V))=0.$$ 

 \vspace{0,3cm}

Since $\gamma$ is positive definite, $W\otimes V^{\perp}=W^{\perp}\otimes V=0$. 

 \vspace{0,3cm}

Hence $Id\otimes Id= W\otimes V+W^{\perp}\otimes V^{\perp},$ which happens only when \begin{center}
$W\otimes V=Id\otimes Id$ $($and $W^{\perp}\otimes V^{\perp}=0)$\ \  or\ \ $W^{\perp}\otimes V^{\perp}=Id\otimes Id$ $($and $W\otimes V=0)$.
\end{center}

Then $\gamma$ is trivially equal to $$(W\otimes V)\gamma(W\otimes V)+(W^{\perp}\otimes V^{\perp})\gamma(W^{\perp}\otimes V^{\perp}).$$

By corollary \ref{corollary2}, item $a)$ implies item $b)$. Therefore both items are equivalent and $\gamma\in CR_{k,m}$.

 \vspace{0,3cm}

Since the set of positive definite states is dense in the set of states and there are states outside $CR_{k,m}$ (See remark \ref{remarkcounterexample}), the set $CR_{k,m}$ is dense in the set of states and it is not closed.
\end{proof}

 \vspace{0,3cm}

The next theorem proves the convexity of $CR_{k,m}$.

\vspace{0,5cm}

\begin{theorem}\label{theoremconvex}The set $CR_{k,m}$ is a dense convex cone within the cone of all states of $\mathcal{M}_k\otimes\mathcal{M}_m$. \end{theorem}
\begin{proof} It is clear from the definition that $CR_{k,m}$ is a cone and, by the previous corollary, it is dense  within the cone of all states. It remains to show its convexity.

 \vspace{0,3cm}

Let $\gamma,\delta\in CR_{k,m}$ and define $\alpha=\gamma+\delta$.  

 \vspace{0,3cm}

 In order to prove that $\alpha\in CR_{k,m}$, our corollary \ref{corollary2} says that given a pair of orthogonal projections
$W\in\mathcal{M}_k$ and $V\in\mathcal{M}_m$  satisfying 
\begin{equation}\label{eq8}
tr(\alpha(W\otimes V^{\perp}))=tr(\alpha( W^{\perp}\otimes V))=0,
\end{equation}

it is sufficient to show that $$\alpha=(W\otimes V)\alpha(W\otimes V)+(W^{\perp}\otimes V^{\perp})\alpha(W^{\perp}\otimes V^{\perp}).$$

 \vspace{0,3cm}

Now, since $\alpha=\gamma+\delta$, $\gamma\geq 0$ and $\delta\geq 0$, we obtain from equation \eqref{eq8} that  \begin{center}
$tr(\gamma(W\otimes V^{\perp}))=tr(\gamma( W^{\perp}\otimes V))=0$ and 
$tr(\delta(W\otimes V^{\perp}))=tr(\delta( W^{\perp}\otimes V))=0$.
\end{center}

 \vspace{0,3cm}

Next, by corollary \ref{corollary2} and the fact that $\{\gamma,\delta\}\subset CR_{k,m}$, we have\begin{center}
 $\gamma=(W\otimes V)\gamma(W\otimes V)+(W^{\perp}\otimes V^{\perp})\gamma(W^{\perp}\otimes V^{\perp})$, $\delta=(W\otimes V)\delta(W\otimes V)+(W^{\perp}\otimes V^{\perp})\delta(W^{\perp}\otimes V^{\perp}).$
\end{center}

 \vspace{0,3cm}

Finally, since $\alpha=\gamma+\delta$, $$\alpha=(W\otimes V)\alpha(W\otimes V)+(W^{\perp}\otimes V^{\perp})\alpha(W^{\perp}\otimes V^{\perp}).$$
\end{proof}

\vspace{0,5cm}

In section 5, we present a couple of extra consequences of corollary \ref{corollary2}  in the form of operations that can be performed on completely reducible states that preserve this property, but before that let us describe  the extreme rays of $CR_{k,m}$ in section 4.

\vspace{0,5cm}

\section{Extreme Rays of $CR_{k,m}$}

\vspace{0,3cm}

In this section we provide  a complete description of the  extreme rays of $CR_{k,m}$ (See proposition \ref{propositionextremerays}), we show that these rays are the ones generated by the separable pure states (separable rank 1 states).

  \vspace{0,3cm}

This proposition shows once more how connected the separability property is to the complete reducibility property. Despite $CR_{k,m}$ being such a big cone (dense inside the set of states, by corollary \ref{corollaryCRdense}), its extreme rays are the same extreme rays of the cone of separable states.

  \vspace{0,3cm}

Before proving this proposition, we need a description of the completely reducible  pure states and another lemma.  \vspace{0,3cm}

  \vspace{0,3cm}

We begin this section by showing that completely reducible pure states  are separable.

  \vspace{0,3cm} 
 
 \begin{lemma}\label{lemmapureCPisseparable} A pure state (rank 1) is completely reducible if and only if it is separable.
 \end{lemma}
 \begin{proof}  Let $\gamma\in\mathcal{M}_k\otimes\mathcal{M}_m$ be a pure state. Then there is complex matrix $R_{k\times m}$ such that 
$$\gamma=(R\otimes Id)u_{m}u_m^*(R^*\otimes Id)=(Id\otimes R^t)u_{k}u_k^*(Id\otimes \overline{R}),$$

 where $u_m\in\mathbb{C}^m\otimes\mathbb{C}^m$ and $u_k\in\mathbb{C}^k\otimes\mathbb{C}^k$ are as in theorem \ref{theoremchoi}.\vspace{0,3cm}

As explained in remark \ref{remarkpropertiesadjoints},
  $F_{\gamma}(Z^t)=RZR^*$ and $G_{\gamma}(X^t)=R^tX\overline{R}$. Hence $$F_{\gamma}\circ G_{\gamma}(X)=RR^*XRR^*.$$

\vspace{0,3cm}

If $\rank(R)>1$ then there are two orthonormal eigenvectors $r_1,r_2$ of $RR^*$ associated to positive eigenvalues $a_1,a_2$, which might be equal.

\vspace{0,3cm}

Notice that $F_{\gamma}\circ G_{\gamma}(r_1r_1^*)=a_1^2r_1r_1^*$ and  $F_{\gamma}\circ G_{\gamma}(r_1r_2^*)=(a_1a_2)r_1r_2^*\neq 0$.

\vspace{0,3cm}

Consider the orthogonal projection $V_1=r_1r_1^*$ and the sub-algebra $V_1\mathcal{M}_kV_1=\{\lambda r_1r_1^*,\lambda\in\mathbb{C}\}$. Notice that $V_1\mathcal{M}_kV_1$ is left invariant by $F_{\gamma}\circ G_{\gamma}: \mathcal{M}_k\rightarrow \mathcal{M}_k$, but there is a matrix $r_1r_2^*\in V_1\mathcal{M}_k V_1^{\perp}+V_1^{\perp}\mathcal{M}_k V_1$  such that $F_{\gamma}\circ G_{\gamma}(r_1r_2^*)\neq 0$. Hence $F_{\gamma}\circ G_{\gamma}: \mathcal{M}_k\rightarrow \mathcal{M}_k$ does not have the complete reducibility property by proposition \ref{propositionnull}.
 
\vspace{0,3cm}
So in order to be completely reducible a pure state must be separable. Finally, since every state that is positive under partial transpose is completely reducible \cite[Theorem 26]{CarielloIEEE}, so is every separable state.
 \end{proof}
 
  \vspace{0,3cm}
 
 \begin{lemma}\label{lemmairreduciblebyimage} Let $\gamma,A\in\mathcal{M}_k\otimes\mathcal{M}_m$ be states such that $\Ima(\gamma)=\Ima(A)$. If there is an orthogonal projection $V_1\in\mathcal{M}_k$ such that\begin{itemize}
 \item  $F_{\gamma}\circ G_{\gamma}(V_1\mathcal{M}_k V_1)\subset V_1\mathcal{M}_k V_1$,
 \item $F_{\gamma}\circ G_{\gamma}|_{V_1\mathcal{M}_kV_1}$ is irreducible,
 \item $F_{\gamma}\circ G_{\gamma}|_{V_1^{\perp}\mathcal{M}_k+\mathcal{M}_kV_1^{\perp}}\equiv 0$,
 \end{itemize}  
 then the same three conditions hold for $F_{A}\circ G_{A}:\mathcal{M}_k\rightarrow \mathcal{M}_k$. In particular, $\gamma$ and $A$ belong to $CR_{k,m}$. 
 \end{lemma}
 \begin{proof}The condition $\Ima(\gamma)=\Ima(A)$ is equivalent to the existence
 of a positive $\epsilon$ such that $A-\epsilon \gamma$ and $\gamma-\epsilon A$ are positive semidefinite. Hence \begin{center}
 $G_{A-\epsilon\gamma}(X)=G_A(X)-\epsilon G_{\gamma}(X)$ and $G_{\gamma-\epsilon A}(X)=G_{\gamma}(X)-\epsilon G_{A}(X)$
 \end{center} are positive semidefinite for every $X\in P_k$. Thus, $\Ima(G_{\gamma}(X))= \Ima(G_{A}(X))$, for every $X\in P_k$.\vspace{0,3cm}

Of course, the same is true for $F_A,F_{\gamma}$. Hence, for every $X\in P_k$, 
\begin{equation}\label{eq3}
\Ima(F_{\gamma}\circ G_{\gamma}(X))= \Ima(F_{A}\circ G_{A}(X)).
\end{equation}

The first and the second items of the statement of this theorem say that  $\Ima(F_{\gamma}\circ G_{\gamma}(V_1))\subset \Ima(V_1)$ and there is no other orthogonal projection $V_2\in V_1\mathcal{M}_k V_1\setminus\{0\}$ such that $\Ima(F_{\gamma}\circ G_{\gamma}(V_2))\subset \Ima(V_2)$. By equation \eqref{eq3}, the same is valid for $F_{A}\circ G_{A}:\mathcal{M}_k\rightarrow \mathcal{M}_k$.
 \vspace{0,3cm}
 
Notice that item 3 of the statement of this lemma corresponds to $F_{\gamma}\circ G_{\gamma}(V_1^{\perp})=0$, since $F_{\gamma}\circ G_{\gamma}:\mathcal{M}_k\rightarrow \mathcal{M}_k$ is completely positive by lemma \ref{lemmaCPselfadjoint}. Again, by equation \eqref{eq3}, the same is true with $F_{A}\circ G_{A}:\mathcal{M}_k\rightarrow \mathcal{M}_k$.
\vspace{0,3cm} 
 
Notice that $F_{\gamma}\circ G_{\gamma}:\mathcal{M}_k\rightarrow \mathcal{M}_k$ and $F_{A}\circ G_{A}:\mathcal{M}_k\rightarrow \mathcal{M}_k$ are completely reducible with $s=1$ and $R= V_1^{\perp}\mathcal{M}_k+\mathcal{M}_kV_1^{\perp}$ in the description of a completely reducible map given in the introduction.
 \end{proof}
 
  \vspace{0,3cm} 
  
\begin{remark}\label{remarkcounterexample}
The last lemma says that the property of  $F_{\gamma}\circ G_{\gamma}:\mathcal{M}_k\rightarrow \mathcal{M}_k$ being irreducible on a unique sub-algebra depends only on the range of $\gamma$. This is not the case with the complete reducibility property in general. For example, let $V_1,V_2$ be orthogonal projections of $\mathcal{M}_k$ satisfying $V_1+V_2=Id$. Let $\gamma=V_1\otimes V_1^t+V_2\otimes V_2^t$. Since $\gamma$ is positive under partial transpose, $\gamma\in CR_{k,k}$. In addition, $u_k\in \Ima(\gamma)\ (u_k$ as defined in theorem \ref{theoremchoi}$)$. Now let $\delta=u_ku_k^*+\epsilon \gamma$, where $\epsilon>0$.
Notice that, by remark \ref{remarkpropertiesadjoints}, $G_{\delta}(X)=X^t+\epsilon G_{\gamma}(X)$. Therefore $G_{\delta}:\mathcal{M}_k\rightarrow \mathcal{M}_k$ is an invertible map for sufficiently small $\epsilon$.  Next, since $tr(\gamma(V_1\otimes V_2^t))=tr(\gamma(V_2\otimes V_1^t))=0$ and $\Ima(\gamma)=\Ima(\delta)$, we also have $tr(\delta(V_1\otimes V_2^t))=tr(\delta(V_2\otimes V_1^t))=0$. By lemma \ref{lemmasubalgebrainvariant}, $V_1\mathcal{M}_kV_1$ is an invariant sub-algebra of $F_{\delta}\circ G_{\delta}:\mathcal{M}_k\rightarrow \mathcal{M}_k$, but $F_{\delta}\circ G_{\delta}|_{V_1\mathcal{M}_kV_2+V_2\mathcal{M}_kV_1}\not\equiv 0$, since $F_{\delta}$ and $G_{\delta}$ are adjoints and they are invertible maps. By proposition \ref{propositionnull}, $\delta$ is not completely reducible despite sharing the same range with $\gamma$.
  \end{remark}

 \vspace{0,3cm}

\begin{proposition}\label{propositionextremerays}The extreme rays of $CR_{k,m}$ are the rays generated by  separable pure states of $\mathcal{M}_k\otimes \mathcal{M}_m$.\end{proposition}
\begin{proof} Let us assume that $\{\lambda \gamma, \lambda>0\}$ is an extreme ray of $CR_{k,m}$.\vspace{0,3cm}

As proved in \cite[Proposition 12]{CarielloIEEE}, if $\gamma\in  \mathcal{M}_k\otimes \mathcal{M}_m$ is completely reducible then $\gamma=\sum_{i=1}^s\gamma_i$, where each $\gamma_i=(V_i\otimes W_i)\gamma (V_i\otimes W_i)$ is non-nulll, $F_{\gamma_i}\circ G_{\gamma_i}:V_i\mathcal{M}_kV_i\rightarrow V_i\mathcal{M}_kV_i$ is irreducible, $F_{\gamma_i}\circ G_{\gamma_i}|_{V_i^{\perp}\mathcal{M}_k+\mathcal{M}_kV_i^{\perp}}\equiv 0$ and $V_iV_j=0$, $W_iW_j=0$ for $i\neq j$. Hence $\gamma_i\in CR_{k,m}$ for every $i$.\vspace{0,3cm}

 Thus, a necessary condition for $\{\lambda \gamma, \lambda>0\}$ to be an extreme ray is $s=1$, i.e., $\gamma=(V_1\otimes W_1)\gamma (V_1\otimes W_1)$, where $F_{\gamma}\circ G_{\gamma}:V_1\mathcal{M}_kV_1\rightarrow V_1\mathcal{M}_kV_1$ is irreducible and $F_{\gamma}\circ G_{\gamma}|_{V_1^{\perp}\mathcal{M}_k+\mathcal{M}_kV_1^{\perp}}\equiv 0$.\vspace{0,3cm}

Now, suppose $\gamma$ is not pure. Then $\gamma=A+B$, where $A,B$ are states such that $\Ima(\gamma)=\Ima(A)=\Ima(B)$ and neither $A$ nor $B$ is a multiple of $\gamma$.\vspace{0,3cm}

By lemma \ref{lemmairreduciblebyimage}, $A,B\in CR_{k,m}$. Hence $\{\lambda \gamma, \lambda>0\}$ is not an extreme ray. Thus, $\gamma$ must be pure. By lemma \ref{lemmapureCPisseparable}, $\gamma$ is separable.\vspace{0,3cm}

Of course, if $\gamma$ is a separable pure state and $\gamma=A+B$, where $A,B$ are states, then $A,B$ are multiples of $\gamma$.  Hence $\{\lambda \gamma, \lambda>0\}$ is an extreme ray of $CR_{k,m}$.
\end{proof}

\vspace{0,5cm}

\section{Operations preserving the complete reducibility property }

\vspace{0,5cm}

In this section we present  operations that can be performed on completely reducible states that preserve this property. The first is taking powers and roots of completely reducible states. The second is performing  partial traces on them. The third is a shuffle of bipartite states resulting in a bipartite state. We also show that this resulting bipartite state  is completely reducible if and only if each state used in the shuffle is completely reducible, which allows the construction of completely reducible states without the properties (i-iii) described  in the introduction.

 \vspace{0,3cm}

 \begin{proposition}\label{prop_power_roots_compred} If $\gamma\in CR_{k,m}$ and $n\in\mathbb{N}$  then $\gamma^n\in CR_{k,m}$ and $\gamma^{\frac{1}{n}}\in CR_{k,m}$. Moreover, if $\delta$ is the orthogonal projection onto the image of $\gamma$ then $\delta\in CR_{k,m}$.
 \end{proposition}
\begin{proof}
First, denote by $A^0$ the orthogonal projection onto the image of $A$. 
 \vspace{0,3cm}

Second, recall that $\gamma^a$ is positive semidefinite, when  $a\in \{0,n,\frac{1}{n}\}$, and the fact that $\ker(\gamma)=\ker(\gamma^a)$.

 \vspace{0,3cm}

Now, let $W\in\mathcal{M}_{k},V\in\mathcal{M}_m$ be orthogonal projections such that 
$$tr(\gamma^a(W\otimes V^{\perp}))=tr(\gamma^a( W^{\perp}\otimes V))=0.$$

 \vspace{0,3cm}

Since $\ker(\gamma)=\ker(\gamma^a)$, we get $tr(\gamma(W\otimes V^{\perp}))=tr(\gamma( W^{\perp}\otimes V))=0.$

 \vspace{0,3cm}

By the complete reducibility property of $\gamma$, we obtain
 $$\gamma=(W\otimes V)\gamma(W\otimes V)+(W^{\perp}\otimes V^{\perp})\gamma(W^{\perp}\otimes V^{\perp}).$$
 
  \vspace{0,3cm}
 
 Hence, \begin{equation}\label{eq9}
 \gamma^a=[(W\otimes V)\gamma(W\otimes V)]^a+[(W^{\perp}\otimes V^{\perp})\gamma(W^{\perp}\otimes V^{\perp})]^a.
 \end{equation}

  \vspace{0,3cm}
 
 Since $W$ and $V$ are orthogonal projections and \begin{itemize}
 \item $\text{Im}([(W\otimes V)\gamma(W\otimes V)]^a)\subset \text{Im}(W\otimes V)$
 \item $\text{Im}([(W^{\perp}\otimes V^{\perp})\gamma(W^{\perp}\otimes V^{\perp})]^a)\subset \text{Im}(W^{\perp}\otimes V^{\perp}),$
 \end{itemize}
 
  \vspace{0,3cm} 
 
we get from equation \ref{eq9} that

  \vspace{0,3cm}

 \begin{itemize}
\item $(W\otimes V)\gamma^a(W\otimes V)=[(W\otimes V)\gamma(W\otimes V)]^a$,
\item $(W^{\perp}\otimes V^{\perp})\gamma^a(W^{\perp}\otimes V^{\perp})=[(W^{\perp}\otimes V^{\perp})\gamma(W^{\perp}\otimes V^{\perp})]^a$.
\end{itemize}.

By adding these last two equations, we obtain from equation \ref{eq9} that $$(W\otimes V)\gamma^a(W\otimes V)+(W^{\perp}\otimes V^{\perp})\gamma^a(W^{\perp}\otimes V^{\perp})=\gamma^a.$$

  \vspace{0,3cm}

By corollary \ref{corollary2}, we have just proved that $\gamma^a\in CR_{k,m}$, when $a\in\{0,n,\frac{1}{n}\}$.
\end{proof}

 \vspace{0,3cm}
 
 The next result concerns partial traces.

 \vspace{0,3cm}

\begin{proposition}\label{proppartialtraceiscompred}Let $\gamma\in \mathcal{M}_{k_1\ldots k_s}$ be a state and suppose that as a bipartite state $\gamma\in \mathcal{M}_{k_{i_1}\ldots k_{i_m}}\otimes \mathcal{M}_{k_{i_{m+1}}\ldots k_{i_s}}$ is completely reducible. Then  the partial trace of $\gamma\in \mathcal{M}_{k_{i_1}\ldots k_{i_m}}\otimes \mathcal{M}_{k_{i_{m+1}}\ldots k_{i_s}}$ on the site $\mathcal{M}_{k_{i_j}}$ $($resulting in a bipartite state$)$ still is completely reducible. \end{proposition}
\begin{proof}
Let $\gamma_1\in \mathcal{M}_{k_{i_2}\ldots k_{i_m}}\otimes \mathcal{M}_{k_{i_{m+1}}\ldots k_{i_s}}$ be the partial trace of $\gamma$  on the site $\mathcal{M}_{k_{i_1}}$. For the other sites $(\mathcal{M}_{k_{i_j}})$ the result is analogous. \vspace{0,3cm}

Let $W\in \mathcal{M}_{k_{i_2}\ldots k_{i_{m}}}$ and  $V\in \mathcal{M}_{k_{i_{m+1}}\ldots k_{i_s}}$ be orthogonal projections such that $$tr(\gamma_1(W\otimes V^{\perp}))=tr(\gamma_1(W^{\perp}\otimes V))=0.$$

 \vspace{0,3cm}

Notice that \begin{center}
$tr(\gamma_1(W\otimes V^{\perp}))=tr(\gamma((Id\otimes W)\otimes V^{\perp}))$ and $tr(\gamma_1(W^{\perp}\otimes V))=tr(\gamma((Id\otimes W^{\perp})\otimes V))$.
\end{center}

 \vspace{0,3cm}

By corollary \ref{corollary2}, $$\gamma=((Id\otimes W)\otimes V)\gamma((Id\otimes W)\otimes V)+((Id\otimes W^{\perp})\otimes V^{\perp})\gamma((Id\otimes W^{\perp})\otimes V^{\perp}).$$

 \vspace{0,3cm}

Now, take the partial trace on the site $\mathcal{M}_{k_{i_1}}$ in this last equation to obtain 
 $$\gamma_1=(W\otimes V)\gamma_1(W\otimes V)+(W^{\perp}\otimes V^{\perp})\gamma_1(W^{\perp}\otimes V^{\perp}).$$
 
  \vspace{0,3cm}
 
 By corollary \ref{corollary2}, it follows that $\gamma_1\in CR_{k_{i_2}\ldots k_{i_m},k_{i_{m+1}}\ldots k_{i_s}}.$
\end{proof}

 \vspace{0,3cm}
 
 Now let us define a shuffle of bipartite states resulting in  bipartite states.

 \vspace{0,3cm}

\begin{definition}\label{definitionshuffle}
Let $\gamma_i=\sum_{j=1}^{n_i}A_j^i\otimes B_j^i\in\mathcal{M}_{k_j}\otimes\mathcal{M}_{m_j}$ for $i=1,\ldots,s$. Define the shuffle of $\gamma_1,\ldots,\gamma_s$ by 
$$S(\gamma_1,\ldots,\gamma_s)=\sum_{j_1,\ldots,j_s}A_{j_1}^{1}\otimes \ldots \otimes A_{j_s}^{s}\otimes B_{j_1}^{1}\otimes \ldots \otimes B_{j_s}^{s}\in \mathcal{M}_{k_1\ldots k_s}\otimes\mathcal{M}_{m_1\ldots m_s}. $$

\end{definition}

 \vspace{0,3cm}

\begin{remark}\label{remarkpropertyshuffle}
We shall need in this work three properties of this shuffle that are not hard to notice. Here are the properties
\begin{itemize}
\item[$1)$]  $S(\gamma_1,\ldots,\gamma_s)=S(\delta_1,\ldots,\delta_s)$ if and only if $\gamma_1\otimes\dots\otimes\gamma_s=\delta_1\otimes\dots\otimes\delta_s$,
\item[$2)$]   $S(\gamma_1,\ldots,\gamma_s)$ is positive semidefinite if and only if $\gamma_1\otimes\ldots\otimes\gamma_s$ is positive semidefinite,
\item[$3)$] 
 $S(\gamma_1,\ldots,\gamma_s)S(\delta_1,\ldots,\delta_s)=S(\gamma_1\delta_1,\ldots,\gamma_s\delta_s)$.
\end{itemize}

\end{remark}

 \vspace{0,3cm}

\begin{theorem}\label{theoremshuffle} Let $\gamma_i\in\mathcal{M}_{k_i}\otimes\mathcal{M}_{m_i}$, for $i=1,\ldots,s$, be states. Then the shuffle $S(\gamma_1,\ldots,\gamma_s)$ is completely reducible as a bipartite state of $ \mathcal{M}_{k_1\ldots k_s}\otimes\mathcal{M}_{m_1\ldots m_s}$  if and only if  each $\gamma_i\in\mathcal{M}_{k_j}\otimes\mathcal{M}_{m_j}$, for $i=1,\ldots,s$, is completely reducible.
\end{theorem}
\begin{proof}
Notice that  a positive multiple of $\gamma_i$ is obtainable from $S(\gamma_1,\ldots,\gamma_s)$ by a successive number of partial traces on the sites $\mathcal{M}_{k_j}$ and $\mathcal{M}_{m_j}$. So by proposition \ref{proppartialtraceiscompred}, every $\gamma_i$ must be completely reducible if $S(\gamma_1,\ldots,\gamma_s)$ is completely reducible. \vspace{0,1cm}

Now, for the converse, assume that each $\gamma_i$ is completely reducible. It is not difficult to see that 
\begin{itemize}
\item
$G_{S(\gamma_1,\ldots,\gamma_s)}=G_{\gamma_1}\otimes\ldots \otimes G_{\gamma_s}:\mathcal{M}_{k_1}\otimes \ldots \otimes\mathcal{M}_{k_s}\rightarrow \mathcal{M}_{m_1}\otimes\ldots \otimes\mathcal{M}_{m_s}$ and
\item 
$F_{S(\gamma_1,\ldots,\gamma_s)}=F_{\gamma_1}\otimes\ldots \otimes F_{\gamma_s}:\mathcal{M}_{m_1}\otimes\ldots \otimes\mathcal{M}_{m_s}\rightarrow \mathcal{M}_{k_1}\otimes \ldots \otimes\mathcal{M}_{k_s}$. 
\end{itemize}

 \vspace{0,3cm}

Therefore, $F_{S(\gamma_1,\ldots,\gamma_s)}\circ G_{S(\gamma_1,\ldots,\gamma_s)}=(F_{\gamma_1}\circ G_{\gamma_1})\otimes\ldots \otimes (F_{\gamma_s}\circ G_{\gamma_s}).$

 \vspace{0,3cm}

By lemma \ref{lemmatensorcompred}, $F_{S(\gamma_1,\ldots,\gamma_s)}\circ G_{S(\gamma_1,\ldots,\gamma_s)}:\mathcal{M}_{k_1\ldots k_s}\rightarrow \mathcal{M}_{m_1\ldots m_s}$ is completely reducible, since we have $F_{\gamma_i}\circ G_{\gamma_i}:\mathcal{M}_{k_i}\rightarrow \mathcal{M}_{k_i}$ completely reducible for every $i$ by the current hypothesis.
\end{proof}

\vspace{0,5cm}

The next lemma describes more properties of the shuffle. 

\vspace{0,5cm}

 \begin{lemma}\label{lemma_properties_shuffle}Let $\gamma_i\in\mathcal{M}_{k_i}\otimes\mathcal{M}_{k_i}$, for $i=1,\ldots,s$, be states. Then the shuffle $S(\gamma_1,\ldots,\gamma_s)\in  \mathcal{M}_{k_1\ldots k_s}\otimes\mathcal{M}_{k_1\ldots k_s}$ satisfies
 \begin{itemize}
 \item[$a)$] $S(\gamma_1,\ldots,\gamma_s)^{\Gamma}=S(\gamma_1^{\Gamma},\ldots,\gamma_s^{\Gamma})$.
 \item[$b)$]  $S(\gamma_1,\ldots,\gamma_s)F=S(\gamma_1 F_1,\ldots,\gamma_s F_s)$, where $F\in \mathcal{M}_{k_1\ldots k_s}\otimes\mathcal{M}_{k_1\ldots k_s}$ is the flip operator  and for each $i$, $F_i\in\mathcal{M}_{k_i}\otimes\mathcal{M}_{k_i}$ is the flip operator. 
 \item[$c)$] $\mathcal{R}(S(\gamma_1,\ldots,\gamma_s))=S(\mathcal{R}(\gamma_1),\ldots,\mathcal{R}(\gamma_s)).$
 \end{itemize}
 \end{lemma} 
  \begin{proof}Let $\gamma_i=\sum_{j=1}^{n_i}A_j^i\otimes B_j^i\in\mathcal{M}_{k_j}\otimes\mathcal{M}_{k_j}$ for $i=1,\ldots,s$. Then 
  $$S(\gamma_1,\ldots,\gamma_s)^{\Gamma}=\sum_{j_1,\ldots,j_s}A_{j_1}^{1}\otimes \ldots \otimes A_{j_s}^{s}\otimes (B_{j_1}^{1})^t\otimes \ldots \otimes (B_{j_s}^{s})^t=S(\gamma_1^{\Gamma},\ldots,\gamma_s^{\Gamma}), $$
  which proves item $a)$.\\
  
For item $b)$,   let $e_{1}^i,\ldots,e_{k_i}^i$ be the canonical basis of $\mathbb{C}^{k_i}$ for $i=1,\ldots,s$. Then $$F_{i}=\sum_{p_i,q_i=1}^{k_i} e_{p_i}^i(e_{q_i}^i)^t\otimes (e_{q_i}^i)(e_{p_i}^i)^t.$$

  Moreover, $F=\sum_{p_1,\ldots,p_s,q_1,\ldots,q_s}e_{p_1^1}e_{q_1^1}^t\otimes\ldots e_{p_s^s}e_{q_s^s}^t\otimes e_{q_1^1}e_{p_1^1}^t\otimes\ldots e_{q_s^s}e_{p_s^s}^t=S(F_1,\ldots,F_s)$.\\

Notice that   $S(\gamma_1,\ldots,\gamma_s)F=S(\gamma_1,\ldots,\gamma_s)S(F_1,\ldots,F_s)=S(\gamma_1F_1,\ldots,\gamma_sF_s),$ by remark \ref{remarkpropertyshuffle}.\\

For item $c)$, there is a result proved in item 7) of \cite[Lemma 2.3]{CarielloArxiv} that says $\mathcal{R}(\gamma^{\Gamma})^{\Gamma}=\gamma F$. Thus,  $\mathcal{R}(\gamma)=(\gamma^{\Gamma} F)^{\Gamma}$. The desired result follows from items $a)$ and $b)$.
  \end{proof}
  
\vspace{0,5cm}  
  
 Now we can produce completely reducible states avoiding the three known requirements described in (i),(ii), (iii) in the introduction.

\vspace{0,5cm}

\section{Completely reducible states avoiding the known conditions}

As discussed in the introduction and shown in corollary \ref{corollaryCRdense}, positive definite states are trivial examples of completely reducible states, which we decide to not exclude them from $CR_{k,m}$ in order to maintain the convexity property of $CR_{k,m}$.

\vspace{0,5cm} 

We also know that there are three types of states that possess this property. So in this section the goal is to construct non trivial examples of completely reducible states  avoiding these three types, which are

\begin{itemize}
\item[$(1)$] positive under partial transpose, i.e., $\delta\geq 0$ and $\delta^{\Gamma}\geq 0$.
\item[$(2)$]  positive under partial transpose composed with realignment, i.e., $\delta\geq 0$ and $\mathcal{R}(\delta^{\Gamma})\geq 0$.
\item[$(3)$]   invariant under realignment, i.e, $\delta\geq 0$ and  $\mathcal{R}(\delta)= \delta$.
\end{itemize}

\vspace{0,5cm}  

We accomplish this goal by proving the following result: if $\gamma_j\in\mathcal{M}_k\otimes \mathcal{M}_k$ is a state that satisfies condition $j\in\{1,2,3\}$ above, but not the others, then $S(\gamma_1,\gamma_2,\gamma_3)\in \mathcal{M}_{k^3}\otimes\mathcal{M}_{k^3}$ is completely reducible, but it does not satisfy any of the conditions $(1),(2)$ and $(3)$ above (See corollary \ref{corollary_new_types}).

\vspace{0,5cm}

We need some preliminary results before presenting this result.

   \vspace{0,5cm} 
   \begin{proposition} Let $\gamma_i\in\mathcal{M}_k\otimes \mathcal{M}_k$ be a state whose image is not contained in the anti-symmetric subspace of $\mathbb{C}^k\otimes\mathbb{C}^k$ for $i=1,\ldots,n$. Then the  following results hold.
\begin{itemize}
\item[$a)$] $\gamma_1^{\Gamma}\otimes\ldots\otimes\gamma_n^{\Gamma}\geq 0$ if and only if  $\gamma_i^{\Gamma}\geq 0$ for $i=1,
\ldots,n$. 
\item[$b)$]$\mathcal{R}(\gamma_1^{\Gamma})\otimes\ldots\otimes\mathcal{R}(\gamma_n^{\Gamma})\geq 0$ if and only if  $\mathcal{R}(\gamma_i^{\Gamma})\geq 0$ for $i=1,\ldots,n$. 
\item[$c)$]$\mathcal{R}(\gamma_1)\otimes\ldots\otimes\mathcal{R}(\gamma_n)=\gamma_1\otimes\ldots\otimes\gamma_n$ if and only if $\mathcal{R}(\gamma_i)=\gamma_i$ for $i=1,\ldots,n$. 
\end{itemize}   
   \end{proposition}
   \begin{proof} Let $u_k\in \mathbb{C}^k\otimes\mathbb{C}^k$ be as in theorem \ref{theoremchoi}, $Id\in\mathcal{M}_k\otimes \mathcal{M}_k$ and $F_k\in\mathcal{M}_k\otimes \mathcal{M}_k$ be the flip operator.
   
    Notice that
   \begin{center}
   $ \mathcal{R}(u_ku_k^*)=Id,\ \mathcal{R}(Id)=u_ku_k^*,\ \mathcal{R}(F_k)=F_k,\ F_{k}^{\Gamma}=u_ku_k^*$ and $(u_ku_k^* )^{\Gamma}=F_k$.
   \end{center}
   
Therefore
    $$(Id+F_k+u_ku_k^*)^{\Gamma}=\mathcal{R}((Id+F_k+u_ku_k^*)^{\Gamma})=\mathcal{R}(Id+F_k+u_ku_k^*)=Id+F_k+u_ku_k^*.$$
   
   \vspace{0,3cm}
   
 Moreover,  $(\cdot)^{\Gamma}$, $\mathcal{R}((\cdot)^{\Gamma})$  and $\mathcal{R}(\cdot)$ are isometries acting on $\mathcal{M}_k\otimes \mathcal{M}_k$ (See the introduction of \cite{CarielloArxiv}). Therefore,
 
 \vspace{0,2cm}
 
\begin{itemize}
\item[(i)] $tr(\gamma_i^{\Gamma}(Id+F_k+u_ku_k^*))=tr(\gamma_i^{\Gamma}((Id+F_k+u_ku_k^*)^{\Gamma})^*)=tr(\gamma_i(Id+F_k+u_ku_k^*)),$\vspace{0,2cm}
\item [(ii)] $tr(\mathcal{R}(\gamma_i^{\Gamma})(Id+F_k+u_ku_k^*))=tr(\mathcal{R}(\gamma_i^{\Gamma})(\mathcal{R}((Id+F_k+u_ku_k^*)^{\Gamma}))^*)=tr(\gamma_i(Id+F_k+u_ku_k^*)),$\vspace{0,2cm}
\item[(iii)]  $tr(\mathcal{R}(\gamma_i)(Id+F_k+u_ku_k^*))=tr(\mathcal{R}(\gamma_i)(\mathcal{R}(Id+F_k+u_ku_k^*))^*)=tr(\gamma_i(Id+F_k+u_ku_k^*)).$\\
\end{itemize}  
 
Next, since   $Id+F_k+u_ku_k^*\geq 0$ and its image is the symmetric subspace of  $\mathbb{C}^k\otimes\mathbb{C}^k$, if $tr(\gamma_i (Id+F_k+u_ku_k^*))=0$ then the image of $\gamma_i$ should be contained in the anti-symmetric subspace of  $\mathbb{C}^k\otimes\mathbb{C}^k$, which is not the case by hypothesis. Hence in items (i), (ii) and (iii) above, we get \begin{center}
$tr(\gamma_i^{\Gamma}(Id+F_k+u_ku_k^*))>0$,  $tr(\mathcal{R}(\gamma_i^{\Gamma})(Id+F_k+u_ku_k^*))>0$ and $tr(\mathcal{R}(\gamma_i)(Id+F_k+u_ku_k^*))>0$.
\end{center}

    \vspace{0,5cm}
   
Now, if  $\gamma_1^{\Gamma}\otimes\ldots\otimes\gamma_n^{\Gamma}\geq 0$ then $$\gamma_1^{\Gamma}\stackrel{>0}{\overbrace{tr(\gamma_2^{\Gamma}(Id+F_k+u_ku_k^*))}}\ldots\ \stackrel{>0}{\overbrace{tr(\gamma_n^{\Gamma}(Id+F_k+u_ku_k^*))}}\geq 0,$$ which implies $\gamma_1^{\Gamma} \geq 0$. Of course the same is valid for every $\gamma_i^{\Gamma}$.\\

Next,  if  $\mathcal{R}(\gamma_1^{\Gamma})\otimes\ldots\otimes\mathcal{R}(\gamma_n^{\Gamma})\geq 0$ then $$\mathcal{R}(\gamma_1^{\Gamma})\stackrel{>0}{\overbrace{tr(\mathcal{R}(\gamma_2^{\Gamma})(Id+F_k+u_ku_k^*))}}\ldots  \stackrel{>0}{\overbrace{tr(\mathcal{R}(\gamma_n^{\Gamma})(Id+F_k+u_ku_k^*))}}\geq 0,$$ which implies $\mathcal{R}(\gamma_1^{\Gamma})\geq 0$. Of course the same is valid for  every $\mathcal{R}(\gamma_i^{\Gamma})$.\\

Finally, if $\mathcal{R}(\gamma_1)\otimes\ldots\otimes\mathcal{R}(\gamma_n)=\gamma_1\otimes\ldots\otimes\gamma_n$, then 
$$\mathcal{R}(\gamma_1)tr(\mathcal{R}(\gamma_2)(Id+F_k+u_ku_k^*))\ldots tr(\mathcal{R}(\gamma_n)(Id+F_k+u_ku_k^*))$$
$$=\gamma_1tr(\gamma_2(Id+F_k+u_ku_k^*))\ldots tr(\gamma_n(Id+F_k+u_ku_k^*)),$$

\vspace{0,2cm}

which implies that $\mathcal{R}(\gamma_1)=\gamma_1$, since $ tr(\mathcal{R}(\gamma_i)(Id+F_k+u_ku_k^*))= tr(\gamma_i(Id+F_k+u_ku_k^*))>0$ for every $i$. Of course $\mathcal{R}(\gamma_i)=\gamma_i$ is also valid for every $i$.

\vspace{0,2cm}

It is obvious that the converse of $a)$, $b)$ and $c)$ are valid. 
   \end{proof}
   
   \vspace{0,5cm}
   
The next theorem says that the three properties discussed at the beginning of this section behave with respect to the shuffle in a similar way to the complete reducibility property  (Compare the next theorem with theorem \ref{theoremshuffle}).
 
    \vspace{0,5cm}
    \begin{theorem}\label{theoremshuffle2} Let  $\gamma_i\in\mathcal{M}_{k_i}\otimes \mathcal{M}_{k_i}$ be a state whose image is not contained in  the anti-symmetric subspace of $\mathbb{C}^k\otimes\mathbb{C}^k$ for $i=1,\ldots,n$. Consider $S(\gamma_1,\ldots,\gamma_n)\in \mathcal{M}_{k_1\ldots k_n}\otimes\mathcal{M}_{k_1\ldots k_n}$. Then 
   \begin{enumerate}
   \item$S(\gamma_1,\ldots,\gamma_n)^{\Gamma}\geq 0$ if and only if  $\gamma_i^{\Gamma}\geq 0$ for every $i$.
   \item  $\mathcal{R}(S(\gamma_1,\ldots,\gamma_n)^{\Gamma})\geq 0$ if and only if  $\mathcal{R}(\gamma_i^{\Gamma})\geq 0$ for every $i$.
      \item  $\mathcal{R}(S(\gamma_1,\ldots,\gamma_n))=S(\gamma_1,\ldots,\gamma_n)$ if and only if  $\mathcal{R}(\gamma_i)=\gamma_i$ for every $i$.
   \end{enumerate}
    \end{theorem}
    \begin{proof}
     First, notice that $S(\gamma_1,\ldots,\gamma_n)^{\Gamma}=S(\gamma_1^{\Gamma},\ldots,\gamma_n^{\Gamma})$ by lemma \ref{lemma_properties_shuffle}.

\vspace{0,3cm}

By remark \ref{remarkpropertyshuffle},  $S(\gamma_1^{\Gamma},\ldots,\gamma_n^{\Gamma})\geq 0$ if and only if $\gamma_1^{\Gamma}\otimes \ldots\otimes\gamma_n^{\Gamma}\geq 0$, which is true if and only if $\gamma_i^{\Gamma}\geq 0$, for $i=1,\ldots,n$, by the previous proposition.
 
\vspace{0,3cm} 
  
  Next,   by lemma \ref{lemma_properties_shuffle}, we have $$\mathcal{R}(S(\gamma_1,\ldots,\gamma_n)^{\Gamma})=S(\mathcal{R}(\gamma_1^{\Gamma}),\ldots,\mathcal{R}(\gamma_n^{\Gamma})).$$
  
\vspace{0,3cm}   
  
  By remark \ref{remarkpropertyshuffle}, $S(\mathcal{R}(\gamma_1^{\Gamma}),\ldots,\mathcal{R}(\gamma_n^{\Gamma}))\geq 0$ if and only if $\mathcal{R}(\gamma_1^{\Gamma})\otimes \ldots\otimes \mathcal{R}(\gamma_n^{\Gamma})\geq 0$, which is  true if and only if 
  $\mathcal{R}(\gamma_i^{\Gamma})\geq 0$, for $i=1,\ldots,n$, by the previous proposition.

  \vspace{0,3cm}
  
  Finally, by lemma \ref{lemma_properties_shuffle}, we have $$\mathcal{R}(S(\gamma_1,\ldots,\gamma_n))=S(\mathcal{R}(\gamma_1),\ldots,\mathcal{R}(\gamma_n)).$$
    
  \vspace{0,3cm}    
    
By remark \ref{remarkpropertyshuffle},     $S(\mathcal{R}(\gamma_1),\ldots,\mathcal{R}(\gamma_n))=S(\gamma_1,\ldots,\gamma_n)$ if and only if $\mathcal{R}(\gamma_1)\otimes\ldots\otimes\mathcal{R}(\gamma_n)=\gamma_1\otimes\ldots\otimes\gamma_n,$ which is  true if and only if    $\mathcal{R}(\gamma_i)= \gamma_i$, for $i=1,\ldots,n$, by the previous proposition.
    \end{proof}

\vspace{0,5cm}  

The next lemma says that the matrices that satisfy conditions (1),(2) and (3) described at the beginning of this section cannot be supported on the anti-symmetric subspace of $\mathbb{C}^k\otimes\mathbb{C}^k$.

\vspace{0,5cm}

\begin{lemma}\label{lemma_image_not_contained}If $\gamma\in \mathcal{M}_k\otimes \mathcal{M}_k$ is a state whose image is contained in the anti-symmetric subspace of $\mathbb{C}^k\otimes\mathbb{C}^k$ then $\gamma$ does not satisfy any of the conditions $(1),(2)$ and $(3)$ described at the beginning of this section.
 \end{lemma}
  \begin{proof}
In this case we have $\gamma F_k=-\gamma$, where $F_k$ is the flip operator in $\mathcal{M}_k\otimes \mathcal{M}_k$.   

 \vspace{0,5cm}  

First, since $F_k^{\Gamma}=u_ku_k^*$,
$$tr(\gamma^{\Gamma}u_ku_k^*)=tr(\gamma^{\Gamma}F_k^{\Gamma})=tr(\gamma F_k)=tr(-\gamma)<0.$$

 Therefore $\gamma^{\Gamma}\ngeq 0$.
 
 \vspace{0,5cm}  

Next, since  $Id=\mathcal{R}(u_ku_k^*)^*$ and the realignment map is an isometry,  $$tr(\mathcal{R}(\gamma^{\Gamma})Id)=tr(\mathcal{R}(\gamma^{\Gamma})\mathcal{R}(u_ku_k^*)^*)=tr(\gamma^{\Gamma}u_ku_k^*)=tr(\gamma F_k)=tr(-\gamma)<0.$$
 
 Therefore $\mathcal{R}(\gamma^{\Gamma})\ngeq 0$.
 
 \vspace{0,5cm}  
 
 Finally, since  $Id=\mathcal{R}(u_ku_k^*)^*$, the realignment map is an isometry and $F_ku_k=u_k$,  $$tr(\mathcal{R}(\gamma) Id)=tr(\mathcal{R}(\gamma)\mathcal{R}(u_ku_k^*)^*)=tr(\gamma u_ku_k^*)=tr(\gamma F_k u_ku_k^*)=tr(-\gamma u_ku_k^*)\leq 0.$$
 
  Therefore $\mathcal{R}(\gamma)\neq \gamma$.
  \end{proof}
  
        \vspace{0,5cm} 

    The next corollary shows how to construct states with the complete reducibility property, but not satisfying any of the conditions  (1),(2) and (3) described at the beginning of this section. This is an interesting result, since these conditions used to be the only known conditions that implied this property.
    
      \vspace{0,5cm}  
 
 \begin{corollary}\label{corollary_new_types} Let $k>2$ and $\gamma_j\in\mathcal{M}_k\otimes \mathcal{M}_k$ be a state for $j=1,\ldots,n$. Let us assume that for each of the three  
 conditions described at the beginning of this section, there is at least one $\gamma_j$ that does not satisfy it, but each $\gamma_j$ satisfies at least one of the three conditions.
 Then $S(\gamma_1,\ldots,\gamma_n)\in \mathcal{M}_{k^n}\otimes\mathcal{M}_{k^n}$ is completely reducible, but it does not satisfy any of the conditions  described at the beginning of this section.
 \end{corollary}
 \begin{proof}
First of all, $S(\gamma_1,\ldots,\gamma_n)$ is a completely reducible state, since $\gamma_1,\ldots,\gamma_n$ are all completely reducible states by \cite[Theorems 26, 27, 28]{CarielloIEEE} and by theorem \ref{theoremshuffle}. 
 
\vspace{0,3cm}

Next, by the previous lemma the images of $\gamma_1,\ldots,\gamma_n$ cannot be subsets of the anti-symmetric subspace of $\mathbb{C}^k\otimes \mathbb{C}^k$.

\vspace{0,3cm}

Therefore, by theorem \ref{theoremshuffle2}, we have  
\begin{itemize}
\item $S(\gamma_1,\ldots,\gamma_n)^{\Gamma}\ngeq 0$, since there is $\gamma_j$ such that $\gamma_j^{\Gamma}\ngeq 0$;
\item $\mathcal{R}(S(\gamma_1,\ldots,\gamma_n)^{\Gamma})\ngeq 0$,  since there is $\gamma_j$ such that $\mathcal{R}(\gamma_j^{\Gamma})\ngeq 0$;
  \item $\mathcal{R}(S(\gamma_1,\ldots,\gamma_n))\neq S(\gamma_1,\ldots,\gamma_n)$,,  since there is $\gamma_j$ such that $\mathcal{R}(\gamma_j)\neq \gamma_j$.
\end{itemize}
 \end{proof}
 
  \vspace{0,5cm}     
 
 \begin{remark}
The reason for asking $k>2$ in this theorem is the  \cite[Lemma 3.34]{Cariello_thesis}. It says that 
every state in $\mathcal{M}_2 \otimes \mathcal{M}_2$ satisfying condition 2 or condition 3 at the beginning of this section also satisfies condition 1. Examples of states in $\mathcal{M}_3 \otimes \mathcal{M}_3$ that satisfy condition 2 or condition 3, but not condition 1, were provided in \cite[Examples 3.36]{Cariello_thesis}. 

It is fairly easy to obtain states satisfying condition 1, but not conditions  2 or 3, since these last two conditions imply a certain symmetry in the Schmidt decomposition of a state $($See \cite[Corollary 25]{CarielloIEEE}$)$.
 \end{remark}

As explained in the end of the introduction, our next and final result is connected to the distillability problem.

\begin{corollary}\label{corollary_distillability}Let $\gamma_i\in \mathcal{M}_k\otimes \mathcal{M}_k$ be a state such that $\mathcal{R}(\gamma_i)=\gamma_i$, for $i=1,\ldots,n$. There isn't $v=a\otimes b+c\otimes d\in \mathbb{C}^{k^n}\otimes \mathbb{C}^{k^n}$
 such that $tr(S(\gamma_1^{\Gamma},\ldots,\gamma_n^{\Gamma})vv^*)<0$ and  $\dim(\text{span}\{a,b,c,d\})=2$. In particular, letting $\gamma_1=\ldots=\gamma_n$ gives the result described in the end of the introduction. \end{corollary}
\begin{proof} The idea is to reduce the problem to $\mathcal{M}_{2}\otimes \mathcal{M}_2$. 

Since $\dim(\text{span}\{a,b,c,d\})=2$, there are $A\in\mathcal{M}_{k^n\times 2}$ and $w\in\mathbb{C}^{2}\otimes \mathbb{C}^2$ such that $(A\otimes A)w=v$.

We shall see that $(A^*\otimes A^t)S(\gamma_1,\ldots,\gamma_n)(A\otimes \overline{A})\in\mathcal{M}_{2}\otimes \mathcal{M}_2$ is positive under partial transpose. 

Then the desired result  follows from lemma \ref{lemma_properties_shuffle}, which says $S(\gamma_1^{\Gamma},\ldots,\gamma_n^{\Gamma})=S(\gamma_1,\ldots,\gamma_n)^{\Gamma}$, together with the following observation:   

 $$tr(S(\gamma_1^{\Gamma},\ldots,\gamma_n^{\Gamma})vv^*)=tr((A^*\otimes A^*)S(\gamma_1,\ldots,\gamma_n)^{\Gamma}(A\otimes A)ww^*)$$
 $$=tr(((A^*\otimes A^t)S(\gamma_1,\ldots,\gamma_n)(A\otimes \overline{A}))^{\Gamma}ww^*)\geq 0.$$
 
\vspace{0,3cm}

Now, in order to complete the proof, we must show that $((A^*\otimes A^t)S(\gamma_1,\ldots,\gamma_n)(A\otimes \overline{A}))^{\Gamma}\geq 0$. It is enough to prove that  $(A^*\otimes A^t)S(\gamma_1,\ldots,\gamma_n)(A\otimes \overline{A})$ is invariant under realignment, since every such state in $\mathcal{M}_{2}\otimes \mathcal{M}_2$ is positive under partial transpose by \cite[Lemma 3.34]{Cariello_thesis}.

By lemma \ref{lemma_image_not_contained}, the image of each $\gamma_i$ is not contained in the anti-symmetric subspace of $\mathbb{C}^{k}\otimes \mathbb{C}^{k}$. Then, by theorem \ref{theoremshuffle2}, $\mathcal{R}(S(\gamma_1,\ldots,\gamma_n))=S(\gamma_1,\ldots,\gamma_n)$.

Next, by item 4 of  \cite[Lemma 23]{CarielloIEEE}, $$\mathcal{R}((A^*\otimes A^t)S(\gamma_1,\ldots,\gamma_n)(A\otimes \overline{A}))=(A^*\otimes A^t)\mathcal{R}(S(\gamma_1,\ldots,\gamma_n))(A\otimes \overline{A}).$$

Finally, since $\mathcal{R}(S(\gamma_1,\ldots,\gamma_n))=S(\gamma_1,\ldots,\gamma_n)$, we obtain  $$\mathcal{R}((A^*\otimes A^t)S(\gamma_1,\ldots,\gamma_n)(A\otimes \overline{A}))=(A^*\otimes A^t)S(\gamma_1,\ldots,\gamma_n)(A\otimes \overline{A}).$$

\end{proof}

  \section*{Summary and Conclusion}
  
    \vspace{0,5cm} 
  
  In this work we studied the set of completely reducible states $(CR_{k,m})$ from different perspectives. First, we rewrote the complete reducibility property in a very straightforward way that allowed us to prove the convexity of $CR_{k,m}$  and to prove that taking powers, roots and partial traces  of completely reducible states result in completely reducible states. Then we showed that a certain Shuffle of bipartite states, which results in a bipartite state, is completely reducible if and only if the  shuffled states are too. This result implied a quite general construction that shows that completely reducible states do not need to be positive under partial transpose, positive under partial transpose composed with realignment or invariant under realignment. These last three conditions were previously the only known conditions that guaranteed  the complete reducibility property for a state, but now we have many ways to produce states with such important property. We completed this paper by showing that  that our results are connected to the distillability problem.

  \section*{Declarations}
  No potential conflict of interest was reported by the author. No funding was received for conducting this study. The author has no relevant financial or non-financial interests to disclose.

   \vspace{0,2cm} 
   
   \section*{Acknowledgement}

The author would like to thank the referee for providing constructive comments and helping in the improvement of this manuscript. 

  \vspace{0,2cm}

\end{document}